\documentclass[letterpaper]{article}

\usepackage{comment}
\usepackage{booktabs}
\usepackage{pifont}
\usepackage{caption}
\usepackage{subcaption}
\usepackage{amssymb,amsmath}
\usepackage{mathtools}
\usepackage{amsthm}
\usepackage{amsfonts}

\usepackage{hyperref}

\newtheorem{definition}{Definition}

\newtheorem{theorem}{Theorem}

\newtheorem{proposition}{Proposition}

\newtheorem{conjecture}{Conjecture}

\newtheorem{example}{Example}

\begin{document}

% In the original styles from ACM, you would have needed to
% add meta-info here. This is not necessary for AAMAS 2015  as
% the complete copyright information is generated by the cls-files.

\title{Random Serial Dictatorship versus Probabilistic Serial Rule: A Tale of Two Random Mechanisms}

\author{
Hadi Hosseini, Kate Larson, Robin Cohen\\
Cheriton School of Computer Science\\
University of Waterloo\\
\{h5hossei,klarson,rcohen\}@uwaterloo.ca
%Kate Larson\\
%Cheriton School of Computer Science\\
%University of Waterloo\\
%klarson@uwaterloo.ca
%\and
%Robin Cohen\\
%Cheriton School of Computer Science\\
%University of Waterloo\\
%rcohen@uwaterloo.ca
}
\date{\today}
\maketitle

\begin{abstract}

For assignment problems where agents, specifying ordinal preferences, are allocated indivisible objects, two widely studied randomized mechanisms are the Random Serial Dictatorship (RSD) and Probabilistic Serial Rule (PS). These two mechanisms both have desirable economic and computational properties, but the outcomes they induce can be incomparable in many instances, thus creating challenges in deciding which mechanism to adopt in practice.
In this paper we first look at the space of lexicographic preferences and show that, as opposed to the general preference domain, RSD satisfies envyfreeness. Moreover, we show that although under lexicographic preferences PS is strategyproof when the number of objects is less than or equal agents, it is strictly manipulable when there are more objects than agents.
In the space of general preferences, we provide empirical results on the (in)comparability of RSD and PS, analyze economic properties, and provide further insights on the applicability of each mechanism in different application domains.

\end{abstract}

% Note that the category section should be completed after reference to the ACM Computing Classification Scheme available at
% http://www.acm.org/about/class/1998/.

%\category{I.2.11}{Distributed Artificial Intelligence}{Multiagent systems}
%\category{J.4}{Social and Behavioral Sciences}{Economics}

%A category including the fourth, optional field follows...
%\category{D.2.8}{Software Engineering}{Metrics}[complexity measures, performance measures]

%General terms should be selected from the following 16 terms: Algorithms, Management, Measurement, Documentation, Performance, Design, Economics, Reliability, Experimentation, Security, Human Factors, Standardization, Languages, Theory, Legal Aspects, Verification.

%\terms{Economics, Theory, Experimentation}

%Keywords are your own choice of terms you would like the paper to be indexed by.

%\keywords{Mechanism Design, Matching, Random Assignment, Probabilistic Serial, Random Serial Dictatorship}

%%%%%%%%%%%%%%%%%%%%%%%%%%%%%%%%%%%%%%%%%%%%%
\section{Introduction}
%%%%%%%%%%%%%%%%%%%%%%%%%%%%%%%%%%%%%%%%%%%%%

The problem of assigning a number of indivisible objects to a set of agents, each with a private preference ordering over the objects, in the absence of monetary transfer, is fundamental in many multiagent resource allocation applications, and has been the center of attention amongst researchers at the interface of artificial intelligence, economics, and mechanism design.
Assigning dormitory rooms or offices to students, students to public schools, college courses to students, teaching load to faculty, organs and medical resources to patients, members to subcommittees, etc. are some of the myriad examples of one-sided matching problems~\cite{roth2004kidney,sonmez2010course,sonmez2010house,ashlagi2013mix,budish2012multi,manlove2013algorithmics}.

%In the context of mechanism design, the random assignment problem with desirable properties such as efficiency, strategyproofness, and fairness has been a pivotal subject in studying matching markets~\cite{manlove2013algorithmics}.
%
%Two frontier randomized matching mechanisms that have attracted attention for their practicality and strong economic properties are \emph{random serial dictatorship} (RSD) and \emph{probabilistic serial rule} (PS).

Two important (randomized) matching mechanisms are \emph{Random Serial Dictatorship} (RSD)~\cite{abdulkadirouglu1998random} and \emph{Probabilistic Serial Rule} (PS)~\cite{bogomolnaia2001new}. Both mechanisms have important economic properties and are practical to implement.
The RSD mechanism has strong truthful incentives but guarantees neither efficiency nor envyfreeness. PS satisfies efficiency and envyfreeness, however, it is susceptible to manipulation. 
%The assignments induced by PS could only be efficient with respect to the non-truthful reports, and hence, pose to be inefficient under true underlying preferences.

%The nuances of PS and RSD mechanisms in inducing random assignment solutions raises a few intriguing questions on justifying their practical usage.
%Despite the efficiency of PS with respect to ordinal preferences, in many instances of preference profiles, the assignments induced by RSD and PS are incomparable, that is, without private underlying utility models of agents it is impossible to assess the two assignments with respect to efficiency.
%Furthermore, although PS does not guarantee strategyproofness, for some instances of preference profiles the assignment induced by PS satisfies strategyproofness, meaning that no agent can benefit from misreporting her preference ordering.
%Similarly, the fairness notion of envy under RSD requires a subtle consideration as in some instances RSD-induced assignments satisfy strict envyfreeness.

Therefore, since both RSD and PS induce random assignments, there are subtle points to be considered when deciding which mechanism to use. For example, given a particular preference profile, the mechanisms often produce random assignments which are simply incomparable and thus, without additional knowledge of the underlying utility models of the agents, it is difficult to determine which is the ``better'' outcome. 
Furthermore, properties like non-manipulability and envyfreeness can depend on whether there is underlying structure in the preferences, and even in general preference models it is valuable to understand under what conditions a mechanism is likely to be non-manipulable or envyfree as this can guide designers choices.

%\subsection{Our Model and Results}

In this paper, we consider the random assignment of $m$ objects to $n$ agents, where each agent reports her complete preferences as a strict ordering over objects. 
We first look at the space of lexicographic preferences and define two axioms for fairness based on partial preference orderings. We show that, as opposed to the general preference domain, RSD satisfies envyfreeness. Moreover, we show that although under lexicographic preferences PS is strategyproof when the number of objects is less than or equal agents, it is susceptible to manipulation when there are more objects than agents.

We empirically study the space of all possible preference profiles for various matching problems, and provide insights on the comparability of RSD and PS under the various agent-object combinations with the aim of providing practical insights on the different properties of RSD and PS.

%Our empirical characterization indicates that in the space of general preferences, when $n \leq m$, RSD is preferable since not only does it provide strategyproof assignments, but RSD's induced assignments are not dominated by the assignments induced by PS. Furthermore, the percentage of weakly envious agents decreases, particularly when $n=m$.
%On the other hand, for $n\geq m$ PS satisfies efficiency and envyfreeness, however, its induced random assignments are almost 100\% (weakly) manipulable, while for $n < m$ the fraction of strictly manipulable assignments converges to 1 as $m - n$ increases.
%%
%Under lexicographic preferences RSD is envyfree, and for $n=m$ the fraction of preference profiles at which RSD is inefficient goes to zero.
%Moreover, for $n\geq m$ PS satisfies efficiency, strategyproofness, and envyfreeness under  lexicographic preferences, but it is highly manipulable when $n < m$.

%We characterize the space of preference profiles at which deploying RSD 

%%%%%%%%%%%%%%%%%%%%%%%%%%%%%%%%%%%%%%%%%%%%%
\section{Relation to the literature}
%%%%%%%%%%%%%%%%%%%%%%%%%%%%%%%%%%%%%%%%%%%%%

%The problem of randomly assigning a set of items to a set of agents captured researchers attentions in an attempt to design alternative efficient mechanisms that provide fair solutions without the use of external mediums such as money.

%It is also interesting to draw the relation between CEEI and ...

%Despite Hylland and Zeckhauser's~\cite{hylland1979efficient} efficient pseudo market based on eliciting cardinal preferences of agents, Gale~\cite{gale1987college} conjectured that there exists no Pareto efficient, strategyproof, and fair (in terms of equal treatment of equals) mechanism in the cardinal domain. Later, Zhou~\cite{zhou1990conjecture} proved this conjecture and further showed that one cannot achieve fairness (in terms of equity) and Pareto efficiency simultaneously when we consider incentive compatible mechanisms.

In the house allocation problem with ordinal preferences, Svensson showed that serial dictatorship is the only deterministic mechanism that is strategyproof, nonbossy, and neutral~\cite{svensson1999strategy}. Abdulkadiroglu and Sonmez showed that uniform randomization over all serial dictatorship assignments, a.k.a RSD, is equivalent to the top trading cycle's core from random initial endowment, and thus satisfies stratgyeproofness, proportionality, and ex post efficiency~\cite{abdulkadirouglu1998random}. Although RSD is strategyproof and proportionally fair, but it does not guarantee envyfreeness.

Bogomolnaia and Moulin noted the inefficiency of the RSD mechanism from the ex ante perspective, and characterized random assignment mechanisms based on first order stochastic dominance (sd)~\cite{bogomolnaia2001new}. They further proposed the probabilistic serial mechanism as an efficient and envyfree mechanism. While PS is not strategyproof, it satisfies weak stratgyeproofness for problems with equal number of agents and objects.
Kojima studied random assignment of multiple indivisible objects and showed that PS is strictly manipulable (not weakly strategyproof) even when there are only two agents~\cite{kojima2009random}.
Kojima and Manea, later, considered a setting with multiple copies of each object and showed that in large assignment problems with sufficienctly many copies of each object, truthtelling is a weakly dominant strategy in PS~\cite{kojima2010incentives}.
In order to give a rationale for practical employment of RSD and PS in real-life applications, Che and Kojima further analyzed the economic properties of PS and RSD in large markets with multiple copies of each object and concluded that these mechanisms become equivalent when the market size becomes large~\cite{che2010asymptotic}. Particularly, they showed that inefficiency of RSD and manipulability of PS vanishes when the number of copies of each object approaches infinity.  

The practical implications of deploying RSD and PS have been the center of attention in many one-sided matching problems~\cite{abdulkadiroglu2009strategy}. 
In the school choice setting with multi-capacity alternatives, Pathak compared RSD and PS using data from the assignment of public schools in New York City and observed that many students obtained a more desirable random assignment through PS~\cite{pathak2006lotteries}. However, the efficiency difference was quite small. These observations further confirm Che and Kojima's equivalence result when there are multiple copies of each object (\textit{i.e.} school seats) available~\cite{che2010asymptotic}. 
Despite these findings for arbitrary large markets, the equivalence results of Che and Kojima, and its extension to all random mechanisms by Liu and Pycia~\cite{liu2013ordinal}, do not hold when the quantities of each object is limited to one~\cite{che2010asymptotic}.

%Motivated by the design of student assignment system in New York City, Pathak and Sethuraman~\cite{pathak2011lotteries} showed the equivalence of single random matching mechanisms (such as RSD) to market-based iterative lottery mechanisms (such as core from random endowment) when there are multiple copies of each object (\textit{i.e.} school seats) available.

%
This paper sets out to study the comparability of PS and RSD when there is only one copy of each object, and analyze the space of all preference profiles for different numbers of agents and objects. 
We make several intriguing observations about the manipulability of PS and fairness properties of RSD. Following Manea's work on asymptotic inefficiency of RSD~\cite{manea2009asymptotic}, we show that despite this inefficiency result, the fraction of random assignments at which PS stochastically dominates RSD vanishes when the number of agents is less than or equal to the available objects. Moreover, we show that this result strongly holds for lexicographic preferences when there is equal number of agents and objects.

%These findings confirm the employment of RSD in many practical settings as unlike PS it is strategyproof and for a class of utility models consistent RSD may potentially be ex-ante efficient.

%%%%%%%%%%%%%%%%%%%%%%%%%%%%%%%%%%%%%%%%%%%%%
\section{Formal Representation}
%%%%%%%%%%%%%%%%%%%%%%%%%%%%%%%%%%%%%%%%%%%%%
%Describing the model
%Defining both RSD and PS

A one-sided matching problem $\langle N, M, \succ \rangle$ consists of a set of agents $N$, where $|N| = n$, a set of distinct indivisible objects $M$ with $|M| = m$, and a \emph{preference profile} $\succ$ denoting the set of strict and complete preference orderings of agents over the objects.
%We assume that there is only one copy of each object.
Let $\mathcal{P}$ denote the set of all complete and strict preferences over $M$. Each agent $i$ has a private preference ordering denoted by $\succ_{i}\in \mathcal{P}$, where $a\succ_{i} b$ indicates that agent $i$ prefers object $a$ over $b$. Thus, a \emph{preference profile} is $\succ = (\succ_{1},\ldots, \succ_{n}) \in \mathcal{P}^{n}$. 
We represent the preference ordering of agent $i$ by the ordered list of objects $\succ_{i} = a \succ_{i} b \succ_{i} c$ or $\succ_{i} = (abc)$, for short.

A random assignment is a stochastic $n\times m$ matrix $A$ that specifies the probability of assigning each object $j$ to each agent $i$. The probability vector $A_{i} = (A_{i,1}, \ldots, A_{i,m})$ denotes the random allocation of agent $i$, that is, 
\begin{equation*}
A = 
\begin{pmatrix}
       A_{1} \\
       A_{2} \\
       \vdots \\
       A_{n} 
\end{pmatrix}
=
\begin{pmatrix}
       A_{1,1} & A_{1,2} & \ldots & A_{1,m} \\
       A_{2,1} & A_{2,2} & \ldots & A_{2,m} \\
       \vdots & \vdots & \ddots & \vdots \\
       A_{n,1} & A_{n,2} & \ldots & A_{n,m} 
\end{pmatrix}
\end{equation*}

Let $\mathcal{A}$ refer to the set of possible assignments. 
An assignment $A \in \mathcal{A}$ is said to be \textit{feasible} if and only if $\forall i\in N, \sum_{j\in M} A_{i,j} = 1$, that is, the probability distribution function is valid for each object. 

Given random assignment $A_{i}$, the probability that agent $i$ is assigned an object that is at least as good as object $\ell$ is defined as follows
\begin{equation}
w(\succ_{i}, \ell, A_{i}) = \sum_{j\in M: j \succeq_{i} \ell} A_{i,j}
\end{equation}

A deterministic assignment is simply a binary matrix of degenerate lotteries over objects that allocates each object to exactly one agent with certainty.
Every random assignments is a convex combination of deterministic assignments and is induced by a lottery over deterministic assignments~\cite{von1953certain}. %Thus, we focus our attention only on random assignments.

A \emph{matching mechanism} is a mapping $\mathcal{M}: \mathcal{P}^{n} \to \mathcal{A} $ from the set of possible preference profiles to the set of random assignments. %A matching mechanism only elicits ordinal preference orderings from agents$\succ_{i}$ over objects.

%A random assignment is a $n\times m$ matrix $\mu: N\times M \to [0,1]$ that assigns object $j$ to agent $i$ with probability $\mu(i,j)$. The set of all feasible random assignments is denoted by $\mathcal{M}$. An assignment $\mu \in \mathcal{M}$ is said to be feasible if and only if $\forall i\in N, \sum_{j\in M} \mu(i,j) = 1$, that is, the probability distribution function is valid for each object.

%A deterministic assignment is a binary matrix $\mu(i,j): N\times M \to \{0,1\}$ that assigns object $j$ to agent $i$ with a degenerate lottery. An assignment is said to be feasible if and only if $\forall i\in N, \sum_{j\in M} \mu_{i,j}(\succ) = 1$, that is, each object gets allocated to exactly one agent.

%%%%%%%%%%%%%%%%%%%%%%%%%%%%%%%%%%%%%%%%%%%%%
\subsection{Properties}
%%%%%%%%%%%%%%%%%%%%%%%%%%%%%%%%%%%%%%%%%%%%%

%Fairness, strategyproofness, envy
%Briefly Recapping the theoretical outcomes from the literature

%In this section, we briefly define the most prominent economic measures such as efficiency, strategyproofness, and envyfreeness with regards to ordinal preferences according to the matching literature.

In this section we define key properties that matching mechanisms should have. In particular, we formally define efficiency, strategyproofness and envyfreeness for (randomized) matching mechanisms.

In the context of deterministic assignments, an assignment $A_{i}$ \emph{Pareto dominates} another assignment $B_{i}$ at $\succ$ if $\exists i\in N$ such that $A_{i} \succ_{i} B_{i}$ and $\forall k\in N$ $A_{k} \succeq_{k} B_{k}$, where $A_{i} \succ_{i} B_{i}$ denotes that agent $i$ strictly prefers assignment $A_{i}$ over $B_{i}$. An assignment is \emph{Pareto efficient} at $\succ$ if no other assignment exists that Pareto dominates it at $\succ$.
Extending the Pareto efficiency requirement to random mechanisms, we focus only on mechanisms that always guarantee a Pareto efficient solution ex post.
A random assignment is called \emph{ex post efficient} if it can be represented as a probability distribution over deterministic Pareto efficient assignments.

%Randomization allows for the allocation of divisible probabilities of obtaining the objects, and relies on ordinal revelation of agents' private preferences. 
%In order to evaluate the quality of the probabilistic assignments, Bogomolnaia and Moulin~\cite{bogomolnaia2001new} proposed an ordinal measure based on first-order stochastic dominance (\emph{sd}) for random allocations.

To evaluate the quality of a random assignment, we follow the convention proposed by Bogomolnaia and Moulin based on first-order stochastic dominance~\cite{bogomolnaia2001new}.

%\subsubsection{Efficiency}

%\begin{definition}
%Given a preference profile $\succ$, assignment $p$ \textbf{stochastically dominates} (SD) assignment $q$ if and only if for each agent $i \in N$
%\begin{gather}
%\forall \ell \in M, \sum_{j \in M: j \succ_{i} \ell} p_{i, j}(\succ) \geq \sum_{j \in M: j \succ_{i} \ell} q_{i,j}(\succ)
%\end{gather}
%\end{definition}

\begin{definition}
Given a preference ordering $\succ_{i}$, random assignment $A_{i}$ \textbf{stochastically dominates} (sd) assignment $B_{i} (\neq A_{i})$ if 
%Random assignment $P$ \textbf{stochastically dominates} (\emph{sd}) assignment $Q (\neq P)$ at $\succ_{i}$ if 
%Given agent $i$'s preference ordering $\succ_{i}$, assignment $P_{i}$ \textbf{stochastically dominates} (sd) assignment $Q_{i} (\neq P_{i})$ if 
\begin{equation}
\forall \ell \in M,\ w(\succ_{i}, \ell, A_{i}) \geq w(\succ_{i}, \ell, B_{i})
\end{equation}
\end{definition}

Stochastic dominance is a strong requirement. It implies that for the entire space of utility functions that is consistent with ordinal preferences, an agent's expected utility under $A$ is always greater than her expected utility under $B$. 

\begin{definition}
A random assignment is \textbf{sd-efficient} if for all agents, it is not stochastically dominated by any other random assignment.
\end{definition}

%The notion of \emph{sd}-efficiency (also known as ordinal efficiency~\cite{bogomolnaia2001new}) is perhaps the most compelling efficiency notion in the context of assignment mechanisms that are solely based on ordinal preferences. 

%
%A matching mechanism is \emph{sd}-efficient if it always induces \emph{sd}-efficient random assignments.

A matching mechanism is \emph{sd-efficient} if at all preference profiles $\succ \in \mathcal{P}^{n}$, for all agents $i\in N$, the induced random assignment is not stochastically dominated by any other assignment.
Intuitively, no other mechanism exists that \textit{all} agents strictly prefer its outcome to their current random assignments.

A matching mechanism is \emph{sd}-strategyproof if there exists no non-truthful preference ordering $\succ'_{i} \neq \succ_{i}$ that improves agent $i$'s random assignment. More formally, 

%for all $i\in N, p(\succ) \succ_{i} p(\succ'_{i},\succ_{-i})$ for all misreports $\succ'_{i} \neq \succ_{i}$. More formally,

\begin{definition} 
Mechanism $\mathcal{M}$ is \textbf{sd-strategyproof} if at all preference profiles $\succ \in \mathcal{P}^{n}$, for all agents $i \in N$, and for any misreport $\succ'_{i} \in \mathcal{P}^{n}$, such that $A = \mathcal{M}(\succ)$ and $A' = \mathcal{M}(\succ'_{i},\succ_{-i})$, we have:
\begin{equation}
\forall \ell \in M,\ w(\succ_{i}, \ell, A_{i}) \geq w(\succ_{i}, \ell, A'_{i})
\end{equation}
\end{definition}

In the context of random mechanisms, \emph{sd}-strategyproofness is a strict requirement. It states that under any utility model consistent with the preference orderings, no agent can improve her expected outcome by misreporting.

We can also define a milder version of strategyproofness. A mechanism is \emph{weakly sd-strategyproof} if at all preference profiles, no agent can misreport her preferences and obtain a random assignment that strictly improves her assignment for the entire space of utility models consistent with her true preference ordering.

\begin{definition}
Mechanism $\mathcal{M}$ is \textbf{weakly sd-strategyproof} if for all preference profiles $\succ\in\mathcal{P}^{n}$, for any agent $i$ with misreport $\succ'_{i}$, where $A = \mathcal{M}(\succ)$ and $A'= \mathcal{M}(\succ'_{i},\succ_{-i})$, we have 
\begin{equation}
\exists \ell\in M, w(\succ_{i}, \ell, A_{i}) > w(\succ_{i}, \ell, A'_{i})
\end{equation}
\end{definition}

%Formally, for all preference profiles $\succ\in\mathcal{P}^{n}$, for any agent $i$ with    misreport $\succ'_{i}$, we have $\exists \ell\in M, w(\succ_{i}, \ell, P_{i}) > w(\succ_{i}, \ell, P'_{i})$, where $P = \mathcal{M}(\succ)$ and $P'= \mathcal{M}(\succ'_{i},\succ_{-i})$.
%
%that stochastically dominates her assignment under the truthful report. 
%Formally, there exists \emph{no} preference profile $\succ\in\mathcal{P}^{n}$ wherein for agent $i\in N$ with misreport $\succ'_{i}$ we have $\forall \ell\in M, w(\succ_{i}, \ell, P'_{i}) \geq w(\succ_{i}, \ell, P_{i})$, where $P = \mathcal{M}(\succ)$ and $P'= \mathcal{M}(\succ'_{i},\succ_{-i})$.

We say that a random assignment induced by mechanism $\mathcal{M}$ is \emph{manipulable} if it is not \emph{sd}-strategyproof. That is, there exists an agent $i$ with preference $\succ_{i}$ that her assignment under a truthful report does not stochastically dominate her assignment under a non-truthful report. Formally, assignment $A$ induced by mechanism $\mathcal{M}$ is \textbf{manipulable} at preference profile $\succ$ if there exists an agent $i\in N$ with misreport $\succ'_{i}$ such that if $A = \mathcal{M}(\succ)$ and $A'= \mathcal{M}(\succ'_{i},\succ_{-i})$, we have $\exists \ell \in M,\ w(\succ_{i}, \ell, A'_{i}) > w(\succ_{i}, \ell, A_{i})$.

%\begin{definition}
%Assignment $P$ induced by mechanism $\mathcal{M}$ is \textbf{manipulable} at preference profile $\succ$ if there exists an agent $i\in N$ with misreport $\succ'_{i}$ such that if $P = \mathcal{M}(\succ)$ and $P'= \mathcal{M}(\succ'_{i},\succ_{-i})$, we have
%\begin{equation}
%\exists \ell \in M,\ w(\succ_{i}, \ell, P'_{i}) > w(\succ_{i}, \ell, P_{i})
%\end{equation}
%\end{definition}

Intuitively, an assignment is manipulable if for some utility function consistent with the ordinal preferences, an agent's expected utility under a non-truthful report improves.

An assignment induced by a matching mechanism is strictly manipulable if the mechanism does not satisfy weak \emph{sd}-strategyproofness, i.e., the induced corresponding random assignment is \emph{sd}-manipulable. Thus, assignment $A$ induced by mechanism $\mathcal{M}$ is \textbf{\emph{sd}-manipulable} at preference profile $\succ$ if there exists an agent $i\in N$ with preference $\succ_{i}$, and a misreport $\succ'_{i}\in\mathcal{P}^{n}$ such that $A'_{i} = \mathcal{M}(\succ'_{i},\succ_{-i})$ stochastically dominates $A_{i} = \mathcal{M}(\succ)$, that is, $\forall \ell \in M,\ w(\succ_{i}, \ell, A'_{i}) \geq w(\succ_{i}, \ell, A_{i})$.

%\begin{definition}
%Assignment $P$ induced by mechanism $\mathcal{M}$ is \textbf{sd-manipulable} at preference profile $\succ$ if there exists an agent $i\in N$ with preference $\succ_{i}$, and a misreport $\succ'_{i}\in\mathcal{P}^{n}$ such that $P'_{i} = \mathcal{M}(\succ'_{i},\succ_{-i})$ stochastically dominates $P_{i} = \mathcal{M}(\succ)$, that is,
%\begin{equation}
%\forall \ell \in M,\ w(\succ_{i}, \ell, P'_{i}) \geq w(\succ_{i}, \ell, P_{i})
%\end{equation}
%\end{definition}

An \emph{sd}-manipulable assignment indicates that there exists an agent that can strictly benefit from misreporting. Clearly, a weakly \emph{sd}-strategyproof assignment is not \emph{sd}-manipulable.

%To analyze the fairness properties of matching mechanisms, we focus attention on the ordinal notion of ex ante envyfreeness. 

To analyze the fairness properties of matching mechanisms, we study the ordinal notion of envyfreeness.
An assignment is \emph{sd}-envyfree if each agent strictly prefers her random allocation to any other agent's assignment, that is, given $\succ_{i}$ agent $i$'s random assignment stochastically dominates any other agent's assignment.

%$P_{i}(\succ_{i})$ stochastically dominates $P_{k}(\succ_{i})$ for all $i, k\in N$. 

\begin{definition}
Given agent $i$'s preference $\succ_{i}$, assignment $A_{i}$ is \textbf{sd-envyfree} if for all agents $\forall k\neq i \in N$, 
\begin{equation}
\forall \ell \in M,\ w(\succ_{i}, \ell, A_{i}) \geq w(\succ_{i}, \ell, A_{k})
\end{equation}
\end{definition}

A matching mechanism satisfies \emph{sd}-envyfreeness if at all preference profiles $\succ \in \mathcal{P}^{n}$, it induces \emph{sd}-envyfree assignments for all agents.

An assignment is \textbf{weakly \emph{sd}-envyfree} if no agent strictly prefers another agent's random assignment to her own.  Formally, for all agents $i,k\in N$ with $A_i \neq A_{k}$, we have $\exists \ell \in M,\ w(\succ_{i}, \ell, A_{i}) > w(\succ_{i}, \ell, A_{k})$.

%Formally, there exists no agent $i\in N$ such that $\forall \ell \in M,\ w(\succ_{i}, \ell, P_{k}) \geq w(\succ_{i}, \ell, P_{i})$, for some $k\in N$.

%%%%%%%%%%%%%%%%%%%%%%%%%%%%%%%%%%%%%%%%%%%%%
\section{Two Random Mechanisms}\label{sec:twoRandom}
%%%%%%%%%%%%%%%%%%%%%%%%%%%%%%%%%%%%%%%%%%%%%

In this section, we formally introduce two widely studied matching mechanisms, Random Serial Dictatorship~\cite{abdulkadirouglu1998random} and Probabilistic Serial rule~\cite{bogomolnaia2001new}.

Random Serial Dictatorship (RSD) is a uniform distribution over all possible (priority) orderings of agents, and for each realization of the orderings, the first agent receives her most preferred object, the next agent receives her most preferred object among the set of remaining objects, and so on until no object remains unassigned.\footnote{For $n < m$,  RSD requires a careful method for picking sequence at each realized priority ordering, which will directly affect the efficiency and envy properties of the assignments~\cite{BL-ECAI14,kalinowski2013strategic}. For simplicity, we assume that when $n < m$, at each priority ordering the first agent chooses $(m - n + 1)$ objects, and the rest of the agents choose one object each.}

\begin{definition}
Given a preference profile, the \textbf{Random Serial Dictatorship (RSD)} mechanism is a uniform distribution over the assignments induced by the set of all possible priority orderings over agents.
\end{definition}

Abdulkadiro{\u{g}}lu and S{\"o}nmez showed the equivalence of the random serial dictatorship and the core from the random endowment in the house allocation problem, and argued that this equivalence justifies the wide use of the RSD mechanism in many practical applications such as student housing, course allocating, etc.~\cite{abdulkadirouglu1998random}.

%The equivalence of the random serial dictatorship and the core from the random endowment in the house allocation problem arguably justifies the wide use of the RSD mechanism in many practical applications~\cite{abdulkadirouglu1998random}.

Bogmolnaia and Moulin proposed the Probabilistic Serial rule (PS)~\cite{bogomolnaia2001new}.
Given a preference profile, the PS mechanism treats objects as a set of divisible objects and simulates a Simultaneous Eating Algorithm (SEA): At every point in time, each agent starts eating from her top choice according to $\succ_{i}$ at the unit speed. %  The eating speed is equal for all agents. 
When an object is completely exhausted (eaten away), each agent eats away from her most preferred objects among the remaining objects. The eating algorithm terminates when all objects are exhausted.

%If there is equal number of agents and objects, then the algorithm terminates when each agent has eaten exactly 1 total unit of objects.

The random assignment of agent $i$, when SEA terminates, is the fraction of each object that has been eaten away by agent $i$. We adopt the following definition from Bogomolnaia and Moulin~\cite{bogomolnaia2001new}:
%The following definition is adopted from~\cite{bogomolnaia2001new}.

\begin{definition} Given a preference profile, the \textbf{Probabilistic Serial rule (PS)} is the random probability assignment by simulating the Simultaneous Eating Algorithm with uniform speed.
\end{definition}

For a given preference profile $\succ$, we let $PS(\succ) \in \mathcal{A}$ and $RSD(\succ)\in \mathcal{A}$ denote the outcomes of PS and RSD mechanisms respectively. 

In their seminal work, Bogomolnaia and Moulin characterized the economic properties of these two mechanisms and showed that RSD does not guarantee \emph{sd}-efficiency~\cite{bogomolnaia2001new}.
The following example, adopted from Bogomolnaia and Moulin (2001), illustrates the inefficiency of the RSD mechanism.
\begin{example}\label{example:rsd-noSD}
Suppose there are four agents $N = \{1,2,3,4\}$ and four objects $M = \{a,b,c,d\}$. Consider the following preference profile $\succ = ((abcd), (abcd), (badc), (badc))$.

\begin{table}[h]
%\scriptsize
\small
\tabcolsep=0.1cm 
% $PS(\succ)=$ 
 \begin{subtable}{0.49\linewidth}
 \centering
 \begin{tabular}{ccccc}
    \hline
     & $a$ & $b$ & $c$ & $d$\\ \hline \hline
    $A_1$ & $1/2$ & $0$ & $1/2$ & $0$ \\
    $A_2$ & $1/2$ & $0$ & $1/2$ & $0$ \\
    $A_3$ & $0$ & $1/2$ & $0$ & $1/2$ \\
    $A_4$ & $0$ & $1/2$ & $0$ & $1/2$ \\
    \hline
  \end{tabular}
  \caption{Assignment under $PS(\succ)$}
\end{subtable}
\begin{subtable}{0.49\linewidth}
\centering  
%   $RSD(\succ) = $ 
   \begin{tabular}{ccccc}
    \hline
     & $a$ & $b$ & $c$ & $d$\\ \hline \hline
    $A_1$ & $5/12$ & $1/12$  & $5/12$ & $1/12$  \\
    $A_2$ & $5/12$ & $1/12$  & $5/12$ & $1/12$  \\
    $A_3$ & $1/12$ & $5/12$ & $1/12$  & $5/12$ \\
    $A_4$ & $1/12$ & $5/12$ & $1/12$  & $5/12$ \\
    \hline
  \end{tabular}
  \caption{Assignment under $RSD(\succ)$}
\end{subtable}%\vspace{-1em}
\caption{Example showing the inefficiency of RSD}
\end{table}
In this example, all agents strictly prefer the assignment induced by PS over the RSD assignment. Thus, RSD is inefficient at this preference profile.
\end{example}

%The aforementioned matching mechanisms have been independently studied for their theoretical properties
% has generated substantial debate regarding the applicability of each of the mechanisms in different domains

%The PS mechanism is a synchronous allocation mechanism as all agents simultaneously compete over the objects with equal eating speeds in the SEA algorithm. On the other hand, in RSD, each agent only competes according to the priority ordering, and thus, RSD resembles an asynchronous mechanism.
%This distinctive nature of RSD and PS results in some contrasting theoretical properties. 

Table~\ref{my-label} summarizes the theoretical results from the literature for both RSD and PS. 
For $n \geq m$, PS satisfies \emph{sd}-efficiency, \emph{sd}-envyfreeness, and weakly \emph{sd}-strategyproofness, while RSD is \emph{sd}-strategyproof, and weakly \emph{sd}-envyfree, but it is not \emph{sd}-efficient nor \emph{sd}-envyfree~\cite{bogomolnaia2001new}. 
However, for $n < m$, PS is \emph{sd}-manipulable, i.e., PS is not even weakly \emph{sd}-strategyproof~\cite{kojima2010incentives}.

\begin{table}%[t]
%\footnotesize
\small
\centering
\begin{tabular}{@{}lllll@{}}
\cmidrule(l){2-5}
 & \multicolumn{2}{c}{$n \geq m$} & \multicolumn{2}{c}{$n < m$} \\ \cmidrule(l){2-5} 
						&      PS     &    RSD      &     PS      &    RSD      \\ \cmidrule(l){2-5}
\emph{sd}-strategyproof &      weak     &    \ding{51}      &     \ding{55}      &    \ding{51}  \\
\emph{sd}-efficiency	&     \ding{51}      &    \ding{55}      &     \ding{51}      &    \ding{55}      \\
\emph{sd}-envyfree		&     \ding{51}      &    weak      &    \ding{51}       &   weak      \\ \bottomrule
\end{tabular}\vspace{-.5em}
\caption{The summary of properties}
\label{my-label}
\end{table}

%%%%%%%%%%%%%%%%%%%%%%%%%%%%%%%%%%%%%%%%%%%%
\subsection{The Incomparability of RSD and PS}
%%%%%%%%%%%%%%%%%%%%%%%%%%%%%%%%%%%%%%%%%%%%

The theoretical properties of RSD and PS do not provide proper insight into the head-to-head comparison and applicability of these two mechanisms.

%RSD does not guarantee sd-efficiency --- meaning that since it does not induce a random assignment equivalent to the ones induced by PS.
%- This property of RSD causes it to sometimes induce an assignment that is guaranteed to be inferior to the PS assignment
%- However, we argue that the fraction of such assignments go to zero -- meaning that in comparison none of the two mechanisms is a clear winner when working with ordinal efficiency notions (i.e. sd-efficiency)

RSD does not guarantee \emph{sd}-efficiency, meaning that the assignments induced by RSD are not equivalent to the \emph{sd}-efficient assignments induced by PS.
However, in many instances of preference profiles the random assignments induced by PS and RSD are simply incomparable, in that neither assignment stochastically dominates the other one. The following example illustrates this subtle distinction even for $n=m=3$.
\begin{example}\label{example:non-comparable}
Suppose there are three agents $N = \{1,2,3\}$ and three objects $M = \{a,b,c\}$. Consider the following preference profile $\succ = ((acb), (abc), (bac))$.

\begin{table}
%\scriptsize
\small
\tabcolsep=0.1cm 
\begin{subtable}{0.49\linewidth}
\centering
% $PS(\succ) = $ 
 \begin{tabular}{ c c c c }
    \hline
     & $a$ & $b$ & $c$ \\ \hline \hline
    $A_{1}$ & $1/2$ & $0$ & $1/2$ \\
    $A_2$ & $1/2$ & $1/4$ & $1/4$ \\
    $A_3$ & $0$ & $3/4$ & $1/4$  \\
    \hline
  \end{tabular}
  \caption{Assignment under $PS(\succ)$}
\end{subtable}
\begin{subtable}{0.49\linewidth}
\centering
%   $RSD(\succ) = $ 
   \begin{tabular}{ c c c c }
    \hline
     & $a$ & $b$ & $c$ \\ \hline \hline
    $A_1$ & $1/2$ & $0$  & $1/2$ \\
    $A_2$ & $1/2$ & $1/6$  & $1/3$ \\
    $A_3$ & $0$ & $5/6$ & $1/6$ \\
    \hline
  \end{tabular}
  \caption{Assignment under $RSD(\succ)$}
\end{subtable}%\vspace{-1em}
\caption{Incomparability of RSD and PS}
\label{tab:incomparability}
\end{table}
As shown in Table~\ref{tab:incomparability}, neither mechanism provides an assignment that stochastically dominates the other: agent 1 receives the same allocation under both RSD and PS, agent 2 strictly prefers PS over RSD, and agent 3 strictly prefers RSD over PS.
Thus, the two assignments are incomparable in terms of sd-efficiency.
%Since the two assignments are incomparable, the efficiency of an allocation (ex ante Pareto efficiency) depends on the underlying vNM utility models of agents.
\end{example}

In such instances, the efficiency of the assignments is ambiguous with respect to ordinal preferences. Thus, the (ex ante) efficiency of the induced assignments is contingent on the underlying von Neumann-Morgenstern (vNM) utility functions.
%
%that is, for some subspace of utility models the RSD-induced assignment satisfies (ex ante) efficiency with respect to the profile of  von  Neumann-Morgenstern (vNM) utility functions, while the PS assignment is efficient for another subspace of utilities.
%
%With $n=m=4$, only a single instance can be found that PS dominates RSD. 
%Looking beyond $n > 4$ we can see that the number of instances at which PS and RSD are incomparable increases. Thus, we are interested in studying the fraction of preference profiles at which RSD assignments are guaranteed to be inefficient.
%

Similarly, the envy of RSD and the manipulability of PS depend on the structure of preference profiles, and thus, a compelling question, that justifies the practical implications of deploying a matching mechanism, is to analyze the space of preference profiles to find the likelihood of inefficient, manipulable, or envious assignments under these mechanisms.

%Similarly, while RSD only satisfies weak \emph{sd}-envyfreeness, in some instances of the preference profiles, RSD satisfies the stronger notion of \emph{sd}-envyfreeness.

%Although PS guarantees \emph{sd}-efficiency, in many instances of preference profiles the random assignments induced by PS and RSD are simply incomparable, in that neither assignment stochastically dominates the other one. 
%In such instances, the efficiency of the assignments is ambiguous with respect to ordinal preferences. Thus, the efficiency of the induced assignments is contingent on the underlying utility functions, that is, for some subspace of utility models the RSD-induced assignment satisfies (ex ante) efficiency with respect to the profile of  von  Neumann-Morgenstern (vNM) utility functions, while the PS assignment is efficient for another subspace of utilities. The following example illustrates this subtle distinction even for  $n=m=3$:

%While RSD does not satisfy \emph{sd}-efficiency, there are many instances of preference profiles at which the assignments induced by RSD and PS are simply non-comparable, i.e., as opposed to Example~\ref{example:rsd-noSD}, neither assignments stochastically dominates the other one.

In the next sections, we focus attention on lexicographic preference, and discuss the properties of RSD and PS in this domain. Then, we empirically study the space of all preference profiles for various matching problems.

%%%%%%%%%%%%%%%%%%%%%%%%%%%%%%%%%%%%%%%%%%%%%
\section{Lexicographic Preferences}\label{sec:lex}
%%%%%%%%%%%%%%%%%%%%%%%%%%%%%%%%%%%%%%%%%%%%%

%Results for RSD and PS with regards to envy
%Here we have two theoretical results on lexicographic preferences

%We adopt the notion of \emph{lexicographic dominance} based on lexical orderings of agents over objects. 
In many real-life scenarios, players have preferences that are lexicographic on alternatives or objects, for example political negotiations, voting problems, team standings in sports like Hockey, sequential screening process for hiring candidates or choosing alternatives, etc.~\cite{fraser1994ordinal,yaman2008democratic,fishburn1974lexicographic}.
Lexicographic preferences are present in various applications and have been extensively studied in artificial intelligence and multiagent systems as a means of assessing allocations based on ordinal preferences~\cite{domshlak2011preferences,saban2013note}.

Given two assignments, an agent prefers the one in which there is a higher probability for   getting the most-preferred object. %That is, the utility of other objects is infinitesimal compared to the top-ranked object. 
Formally, given a preference ordering $\succ_{i} = o^1 \succ o^{2}\ldots \succ o^{m}$, agent $i$ prefers any allocation $A_{i}$ that assigns a higher probability to her top ranked object $p_{i,o^{1}}$ over any assignment $B_{i}$ with $B_{i,o^{1}} < A_{i,o^{1}}$, regardless of the assigned probabilities to all other objects.
Only when two assignments allocate the same probability to the top object will the agent considers the next-ranked object. %For two random assignments we define the following relation:

%With lexicographic preferences, agents always prefer an amount of one object to any amount of other objects. That is, the utility of other objects is infinitesimal compared to the top-ranked object. Only when two assignments allocate the same amount of the top object will the agent considers the next-ranked object.

%Lexicographic preferences have been in the center of attention by many researchers as a mean of assessing outcomes based on ordinal preferences~\cite{schulman2012allocation}\cite{saban2013note}, and thus, can provide a new way of comparing random assignments.

%We define the following relation based on lexicographic preferences between two random assignments.

\begin{definition}
Given agent $i$'s preference ordering $\succ_{i}$, assignment $A_{i}$ \textbf{lexicographically dominates (ld)} assignment $B_{i}$ if
\begin{gather}
\exists\ \ell \in M: A_{i,\ell} > B_{i, \ell}\ \wedge\ \forall j \succ_{i} \ell:  A_{i,j} = B_{i, j}.
\end{gather} 

%\begin{gather}
%\exists \ell \in M, W(\succ_{i}, \ell, P_{i}) > W(\succ_{i}, \ell, Q_{i}) \wedge\\ 
%\forall j \succ_{i} \ell, W(\succ_{i}, j, P_{i}) = W(\succ_{i}, j, Q_{i})
%\end{gather} 
 
\end{definition}

A matching mechanism is lexicographically efficient (\emph{ld-efficient}) if for all preference profiles its induced assignment is not lexicographically dominated by any other random assignment. 
By definition \emph{sd}-efficiency yields \emph{ld}-efficiency, however, an \emph{ld}-efficient assignment may not be \emph{sd}-efficient. The following example illustrates this:

\begin{example}\label{example:ld-domination}
Consider four agents $N = \{1,2,3,4\}$ and four objects $M = \{a,b,c,d\}$ at the following preference profile $\succ = ((cabd), (acdb), (cbda), (acbd))$. %RSD and PS induce the following random assignments:
\begin{table}[b]
%\scriptsize
\small
\tabcolsep=0.1cm 
\begin{subtable}{0.49\linewidth}
 \centering
%$PS(\succ) = $ 
\begin{tabular}{ c c c c c}
    \hline
     & $a$ & $b$ & $c$ & $d$\\ \hline \hline
    $A_1$ & $0$ & $1/3$ & $1/2$ & $1/6$ \\
    $A_2$ & $1/2$ & $0$ & $0$ & $1/2$ \\
    $A_3$ & $0$ & $1/3$ & $1/2$ & $1/6$ \\
    $A_4$ & $1/2$ & $1/3$ & $0$ & $1/6$ \\
    \hline
  \end{tabular}
   \caption{$PS(\succ)$ assignment}
\end{subtable}
\begin{subtable}{0.49\linewidth}
 \centering
%   $RSD(\succ) = $ 
   \begin{tabular}{ c c c c c}
    \hline
     & $a$ & $b$ & $c$ & $d$\\ \hline \hline
    $A_1$ & $1/12$ & $1/3$  & $5/12$ & $1/6$  \\
    $A_2$ & $11/24$ & $0$  & $1/12$ & $11/24$  \\
    $A_3$ & $0$ & $5/12$ & $5/12$  & $1/6$ \\
    $A_4$ & $11/24$ & $1/4$ & $1/12$  & $5/24$ \\
    \hline
  \end{tabular}
   \caption{$RSD(\succ)$ assignment}
\end{subtable} %\vspace{-1em}
\caption{An example showing PS dominating RSD lexicographically but not w.r.t stochastic dominance.}
\label{tab:ld}
\end{table}
Table~\ref{tab:ld} shows the assignments induced by PS and RSD.
Here, PS lexicographically dominates RSD since all the agents receive a higher probability of their more preferred objects under PS. However, PS does not stochastically dominate RSD because agent 2 (similarly agent 4) weakly prefers the RSD assignment as $w(\succ_{2}, c,RSD(\succ)) = \frac{11}{24} + \frac{1}{12} = \frac{13}{24}$ is greater than $w(\succ_{2}, c, PS(\succ)) = \frac{1}{2} + 0 = \frac{1}{2}$.
Thus, the two random assignments are in fact incomparable with respect to stochastic dominance.
\end{example}

\begin{proposition}
PS is ld-efficient, RSD is not.
\end{proposition}
\begin{proof}
The \emph{ld}-efficiency of PS is directly derived from its \emph{sd}-efficiency. Example \ref{example:rsd-noSD} illustrates that RSD is not \emph{ld}-efficient.
\end{proof}

\emph{Sd}-efficiency implies \emph{ld}-efficiency. However, similar to stochastic dominance relations, in terms of matching mechanisms it is unclear whether PS or RSD are comparable in terms of \emph{ld}-efficiency, in particular at those instances of preference profiles that  they are non-comparable under stochastic dominance.

%Under ordinal preferences, the strongest concept of efficiency is \emph{sd}-efficiency, which is still weaker than ex ante Pareto efficiency when eliciting vNM utility functions. The following relation shows how \emph{ld}-efficiency is situated within the aforementioned efficiency axioms: ex ante Pareto efficiency $>$ \emph{sd}-efficiency $>$ \emph{ld}-efficiency.

%\begin{figure}[h!]
%\centering
%\includegraphics[scale=.8]{plots/scale.pdf}
%\end{figure}

\subsection{Envyfreeness}

In many practical situations, such as political negotiations, hiring candidates, and sports standings~\cite{fraser1994ordinal,fishburn1974lexicographic}, players most often associate higher regard to their top choices, and assess the fairness of random outcomes according to \emph{ld}-envyfreeness.
%Therefore, we define the \emph{ld}-envyfreeness property based on the notion of lexicographic dominance.

\begin{definition}
Given agent $i$'s preference $\succ_{i}$, assignment $A_{i}$ is \textbf{ld-envyfree} if there exists no agent-object pair $k\in N,\ell\in M$ such that,
$
A_{k,\ell} > A_{i, \ell}\ \wedge\ \forall j \succ_{i} \ell:  A_{i,j} = A_{k,j}
$.
\end{definition}

A matching mechanism satisfies \emph{ld}-envyfreeness if at all preference profiles $\succ\in\mathcal{P}^{n}$ it induces \emph{ld}-envyfree assignments for all agents.
\emph{Sd}-envyfreeness leads to \emph{ld}-envyfreeness; however, \emph{ld}-envyfree allocations may not satisfy \emph{sd}-envyfreeness. 
%The notion of \emph{sd}-envyfreeness leads to \emph{ld}-envyfreeness. Nonetheless, \emph{ld}-envyfree allocations may not satisfy \emph{sd}-envyfreeness. 

\begin{example}  
In Example \ref{example:non-comparable}, under RSD no agent is ld-envious of any other agent. However, agent 2 is weakly envious of agent 3's assignment. Thus, the assignment is not sd-envyfree. 
Assume agent 2's utility for the objects is $u_{2}(a) > u_{2}(b) > u_{2}(c)$. Agent 2's utility for objects as good as object $b$ under her assignment and agent 3's assignment can be written as follows
$
\frac{5}{6} u_{2}(b) > \frac{3}{6} u_{2}(a) + \frac{1}{6} u_{2}(b)
$. 
Thus, for all utility functions at which $u_{2}(b) > \frac{3}{4} u_{2}(a)$, agent 2 prefers agent 3's random assignment.
\end{example}

Based on the definition of stochastic dominance, it is easy to see that every \emph{sd}-envyfree mechanism satisfies \emph{ld}-envyfreeness. PS satisfies \emph{sd}-envyfreeness,  hence, it is \emph{ld}-envyfree for all preference profiles. However, \emph{ld}-envyfreeness is stronger than weak \emph{sd}-envyfreeness.

\begin{proposition}\label{prop:ld-inclusion}
Every ld-envyfree allocation is weakly sd-envyfree, that is, ld-envyfree $\subset$ weakly sd-envyfree.
\end{proposition}

The inclusion in Proposition \ref{prop:ld-inclusion} is strict: consider two agents with $\succ_{1}=\succ_{2}=(abc)$. The random assignments $A_{1}=(\frac{1}{6},\frac{5}{6},0)$ and $A_{2}=(\frac{4}{6},\frac{1}{6},\frac{1}{6})$ are weakly \emph{sd}-envyfree, but do not satisfy \emph{ld}-envyfreeness.
%1/6, 5/6, 0/6
%4/6, 1/6, 1/6
The relationship between the different definitions of fairness is as follows:\footnote{For a complete discussion on the computational complexity of finding or verifying the fairness concepts with respect to envyfreeness and proportionality see \cite{Aziz:2014}.} %ex post envyfree $>$ \emph{sd}-envyfree $>$ \emph{ld}-envyfree $>$ \emph{sd}-proportionality $>$ ex ante envyfree.
%\vspace{-1em}
\begin{figure}[h!]
\centering
\includegraphics[scale=.9]{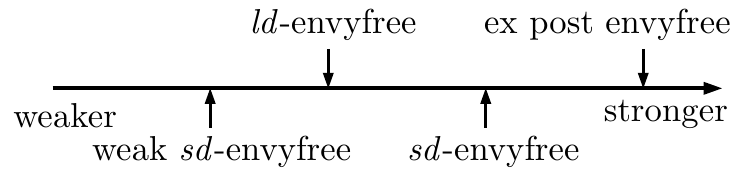}
\end{figure}
%\vspace{-1em}

%\begin{definition}
%Let $\succ_{i}^{\ell}$ denote the partial preference of agent $i$, such that $\succ_{i}^{-\ell}= o^{\ell} \succ_{i} \ldots \succ_{i} o^{m}$.
%A random assignment satisfies upward partial symmetry (UPS) if for all agent $i,k \in N$,
%\begin{equation}
%\forall \ell~ \text{where}~ \succ_{i}^{\ell} = \succ_{k}^{\ell}~ \text{then}~ \forall \ell' \leq \ell, p_{i,\ell'} = p_{k,\ell'}
%\end{equation}
%\end{definition}

Before analyzing the envyfreeness of RSD in the lexicographic domain, we need to define a simple axiom called \emph{downward partial symmetry} (DPS). Intuitively, DPS is an extension to the fairness notion of equity (a.k.a proportionality). DPS states that starting from the first-ranked objects downwards, all agents with identical partial preferences receive exactly the same random assignment of the objects in the partial ordering.

\begin{definition}
Let $\succ_{i}^{\ell}$ denote the partial preference of agent $i$, such that $\succ_{i}^{\ell}= o^{1} \succ_{i} \ldots \succ_{i} o^{\ell}$.
A random assignment satisfies \textbf{downward partial symmetry (DPS)} if for all agent $i,k \in N$, 
\begin{equation}
\forall \ell~ \text{where}~ \succ_{i}^{\ell} = \succ_{k}^{\ell}~\Leftrightarrow~\text{for all}~ j \succeq_{i} \ell, A_{i,j} = A_{k,j}
\end{equation}
\end{definition}

Similarly, a random assignment satisfies \textbf{upward partial symmetry (UPS)} if starting from the \emph{least favorite} objects, agents with identical partial preferences receive exactly the same random assignment of the objects in the partial ordering.
It is easy to see that PS satisfies both UPS and DPS. However, RSD only satisfies DPS.

\begin{proposition}\label{prop:DPS}
RSD satisfies DPS but does not guarantee UPS.
\end{proposition}
\begin{proof}
In Example~\ref{example:ld-domination}, agent 1 and 4's assignments illustrate the non-existence of UPS.

RSD is a uniform distribution over the set of priority orderings, and thus, satisfies equal treatment of equals for full preference orderings, i.e., $\succ_{i} = \succ_{k}$ then $A_{i,\ell} = A_{k,\ell}, \forall \ell$. %One interpretation of RSD's proportionality is that all agents have exactly the same chance of choosing the objects that are attributed equal rankings in their preference orderings.
Assume for contradiction that for two agents $i,k\in N$ where $\succ_{i}^{\ell} = \succ_{k}^{\ell}$, $\exists\ \ell' \succeq_{i} \ell$ such that $A_{i,\ell'} \neq A_{k,\ell'}$.
For the sake of simplicity, let us assume $A_{i,\ell'} < A_{k,\ell'}$. 
For any serial dictatorship where $i$ precedes $k$ in the priority ordering, $i$ chooses an object $j$ such that $j\succ_{i} \ell'$. This immediately implies that $A_{i,j} > A_{k,j}$ for some $j \succ_{k} \ell'$. We can continue this by induction backward up to the first-ranked item, implying that there exists a priority ordering at which agent $i$ chooses a less preferred object, that is $\succ_{i}^{\ell} \neq \succ_{k}^{\ell}$, which contradicts the assumption.
Similarly we can show that if $\succ_{i}^{\ell} \neq \succ_{k}^{\ell}$, it is impossible to achieve $A_{i,j} = A_{k,j}, \forall j \succeq_{i} \ell$.
\end{proof}

The next theorem shows that although RSD is not \emph{sd}-envyfree, it satisfies the weaker notion of \emph{ld}-envyfreeness.

\begin{theorem}
RSD is ld-envyfree, for any $n$ and $m$.
\end{theorem}

\begin{proof}
Assume for contradiction that there exists an agent $k\in N$ with random assignment $A_{k}$ that agent $i$ \emph{ld} prefers $A_{k}$ to her own $A_{i}$. Therefore, by definition of lexicographic dominance, there exists an object $\ell \in M$ such that RSD assigns higher probability to $k$, \textit{i.e.} $A_{i,\ell} < A_{k,\ell}$, and $\forall j\in M: j\succ_{i} \ell$, $A_{i,j} = A_{k,j}$. According to Proposition~\ref{prop:DPS} this is only possible when $\succ_{i}^{\ell} \neq \succ_{k}^{\ell}$, implying that RSD is \emph{ld}-envyfree.
\end{proof}

The \emph{ld}-envyfreeness of RSD is noteworthy: it shows that despite RSD is not \emph{sd}-envyfree, it satisfies envyfreeness with respect to lexicographic preferences.

\subsection{Strategyproofness}

When multiple objects can be allocated to each agent, manipulating the PS outcome may result in a stochastically dominant assignment. In fact, even for $n = 2$, the random assignment induced PS is \emph{sd}-efficient and \emph{sd}-envyfree but not weakly \emph{sd}-strategyproof (i.e. \emph{sd}-manipulable)~\cite{kojima2009random}.

\begin{example}\label{example:sd-manipulable}
Consider two agents and four objects with preferences $\succ = ((abcd), (bcad))$. Assume that agent 1 misreports her preference as $\succ'_{1} = (bacd)$. 
\begin{table}[h!]
%\scriptsize
\small
\tabcolsep=0.1cm
\centering
\begin{subtable}{0.49\linewidth}
 \centering
%$PS(\succ_1, \succ_{2}) = $ 
	\begin{tabular}{ c c c c c}
    \hline
     & $a$ & $b$ & $c$ & $d$\\ \hline \hline
    $A_1$ & $1$ & $0$ & $1/2$ & $1/2$ \\
    $A_2$ & $0$ & $1$ & $1/2$ & $1/2$ \\
    \hline
  \end{tabular}
  \caption{$PS(\succ_1, \succ_{2})$ assignment}
  \label{tab:truth}
\end{subtable}\hfill
\begin{subtable}{0.49\linewidth}
 \centering
%   $PS(\succ'_{1}, \succ_{2}) = $ 
   \begin{tabular}{ c c c c c}
    \hline
     & $a$ & $b$ & $c$ & $d$\\ \hline \hline
    $A_1$ & $1$ & $1/2$  & $0$ & $1/2$  \\
    $A_2$ & $0$ & $1/2$  & $1$ & $1/2$  \\
     \hline
  \end{tabular}
  \caption{$PS(\succ'_{1}, \succ_{2})$ assignment}
  \label{tab:nontruth}
\end{subtable}%\vspace{-1em}
\caption{Random assignments under~(a) truthful and~(b) non-truthful reports.}
\end{table}

Agent 1's allocation under the misreport stochastically dominates her assignment when reporting truthfully because for all objects $\ell$, $w(\succ_{1}, \ell, PS(\succ'_{1}, \succ_{2})) \geq w(\succ_{1}, \ell, PS(\succ_{1}, \succ_{2}))$.
\end{example}

%A \emph{sd}-manipulable assignment assumes that there is an agent that can misreport her preference such that the her induced random allocation stochastically dominates her allocation when reporting truthfully.

The main intuition for the strong manipulablity of PS comes from the eating sequence of agents. Since agents can be allocated more than one object, the sequence in which agents choose to eat away the objects becomes crucial (for a discussion on the game-theoretic properties and complexity of manipulation in picking sequences see \cite{kalinowski2013strategic,BL-ECAI14,bouveret2011general}). 

It is apparent that \emph{sd}-manipulablity leads to \emph{ld}-manipulatiy, however, there are random assignments under PS that are \emph{ld}-manipulable but not \emph{sd}-manipulable, as demonstrated by the following inclusion relation: \emph{sd}-manipulable $\subset$ \emph{ld}-manipulable $\subset$ manipulable.
%Consequently, based on the definitions of strategyproofness one can consider the following incentive compatibility relation: \emph{sd}-strategyproof $>$ \emph{ld}-strategyproof $>$ weakly \emph{sd}-strategyproof.

In the domain of divisible objects, PS is the only stochastic mechanism which is \emph{sd}-efficient, \emph{sd}-envyfree, and \emph{ld}-strategyproof on the lexicographic preference domain~\cite{schulman2012allocation}. However, this characterization is only valid for $n \geq m$.

%For $n \leq m$ PS is ld-strategyproof but not sd-strategyproof.
\begin{proposition}
For $n < m$, PS is not ld-strategyproof nor sd-strategyproof.
\end{proposition}

\begin{proof}
Proof follows from Example~\ref{example:sd-manipulable}.
\end{proof}

\emph{Ld}-strategyproofness is the direct consequence of a simple axiom called \emph{bounded invariance property}~\cite{bogomolnaia2012probabilistic}, that is, no agent is able to change her preference order for any object she likes less than $j$ such that her (probabilistic) allocation of $j$ improves. This axiom was first characterized by Bogomolnaia and Heo, where they showed that PS satisfies bounded invariance for $n=m$~\cite{bogomolnaia2012probabilistic}. However, since for $n < m$ each agent has an eating capacity of more than 1, PS no longer satisfies the bounded invariance property.

\begin{proposition}
For $n < m$, PS does not satisfy the bounded invariance property.
\end{proposition}

%%%%%%%%%%%%%%%%%%%%%%%%%%%%%%%%%%%%%%%%%%%%%
\section{Experimental Results}
%%%%%%%%%%%%%%%%%%%%%%%%%%%%%%%%%%%%%%%%%%%%%

The theoretical properties of PS and RSD, even in restricted domains such as lexicographic preferences, only provide limited insight into their practical applications. In particular, when deciding which mechanism to use in different settings, the incomparability of PS and RSD leaves us with an ambiguous choice. Thus, we examine the properties of RSD and PS in the space of all possible preference profiles.

%We examine the comparability of the defined economic properties for PS and RSD mechanisms. Studying the behavior of each of these mechanisms at non-comparable preference profiles opens up new practical insights in deciding which assignment mechanism to use in different decision problems.

The number of all possible preference profiles is super exponential $(m!)^{n}$. Thus, 
for each combination of $n$ agents and $m$ objects we considered 1,000 randomly generated instances by sampling from a uniform preference profile distribution. Thus, each data point is the average of 1,000 replications. For each preference profile, we ran both PS and RSD mechanisms and compared their outcomes (in terms of random assignments). 

%Our experimentation disclose several intriguing observations, confirming the theoretical results and providing additional insights into matching markets.
A preliminary look at our empirical results illustrates the followings: when $m \leq 2, n \leq 3$, PS coincides exactly with RSD, which results in the best of the two mechanisms, i.e., both mechanisms are \emph{sd}-efficient, \emph{sd}-strategyproof, and \emph{sd}-envyfree.\footnote{This was first noted by Bogomolnaia and Moulin for $n = m = 2$~\cite{bogomolnaia2001new}.}
Moreover, when $m = 2$, for all $n$ PS is \emph{sd}-strategyproof (although the PS assignments are not necessarily equivalent to assignments induced by RSD), RSD is \emph{sd}-envyfree, and for most instances PS stochastically dominates RSD, particularly when $n \geq 4$. 
%Moreover, when $m = 2$, RSD is \emph{sd}-envyfree at all preference profiles and for any number of agents.

%\item When the fraction $\frac{m}{n}$ is small, objects are scarce, \textit{i.e.}, the demand is strictly higher than supply, PS is less manipulable. In fact, the closer $\frac{m}{n}$ gets to $0$, PS becomes more efficient than RSD while an allocation induced by RSD would cause less (weak) envy among the population of agents.

\subsection{Efficiency}

\begin{figure*}[h]
\centering
\begin{subfigure}[b]{0.49\textwidth}
				\includegraphics[width=\textwidth]{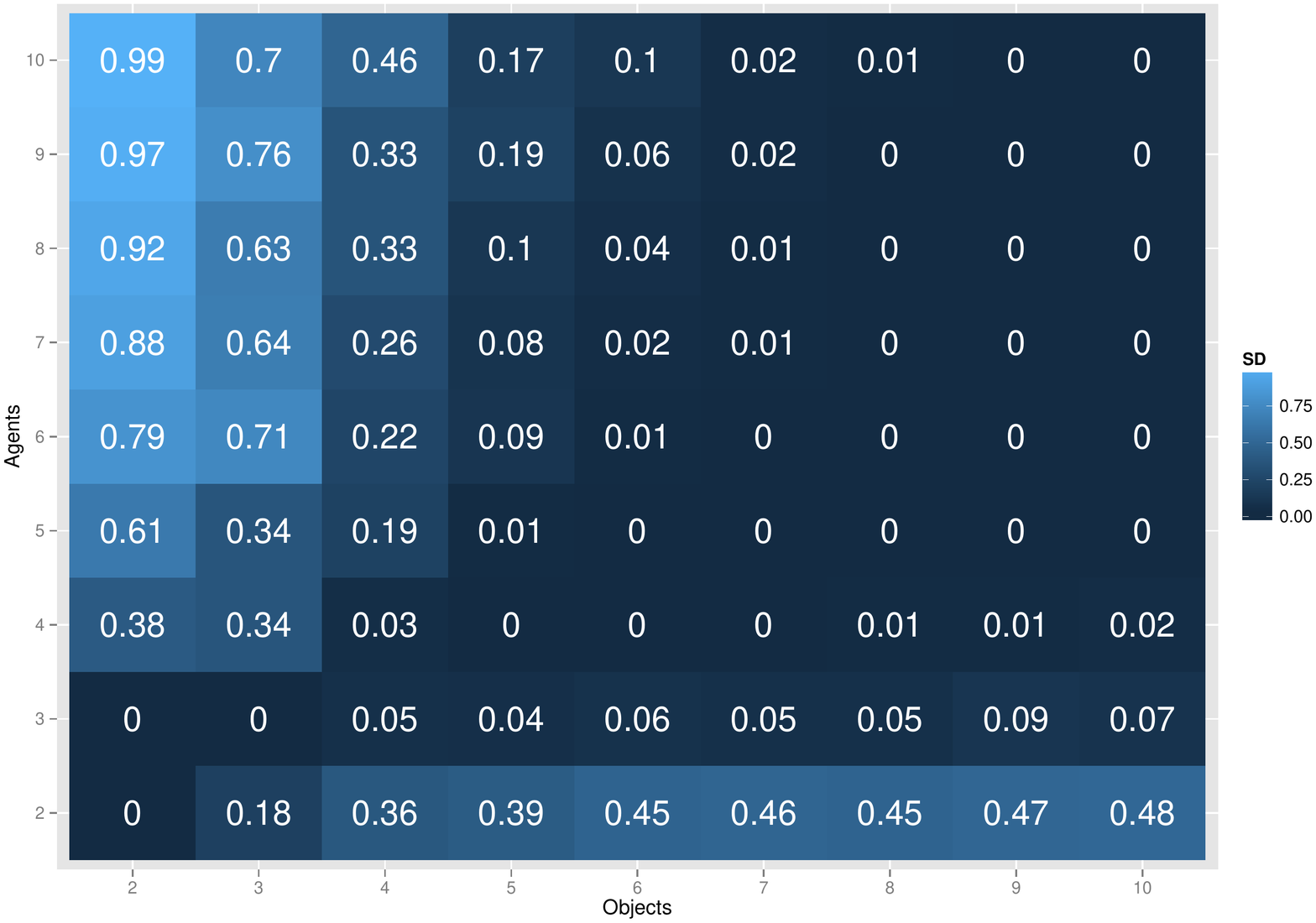}
                \caption{The fraction that PS stochastically dominates RSD.}
                \label{fig:sd-efficiency}
\end{subfigure}
\begin{subfigure}[b]{0.49\textwidth}
				\includegraphics[width=\textwidth]{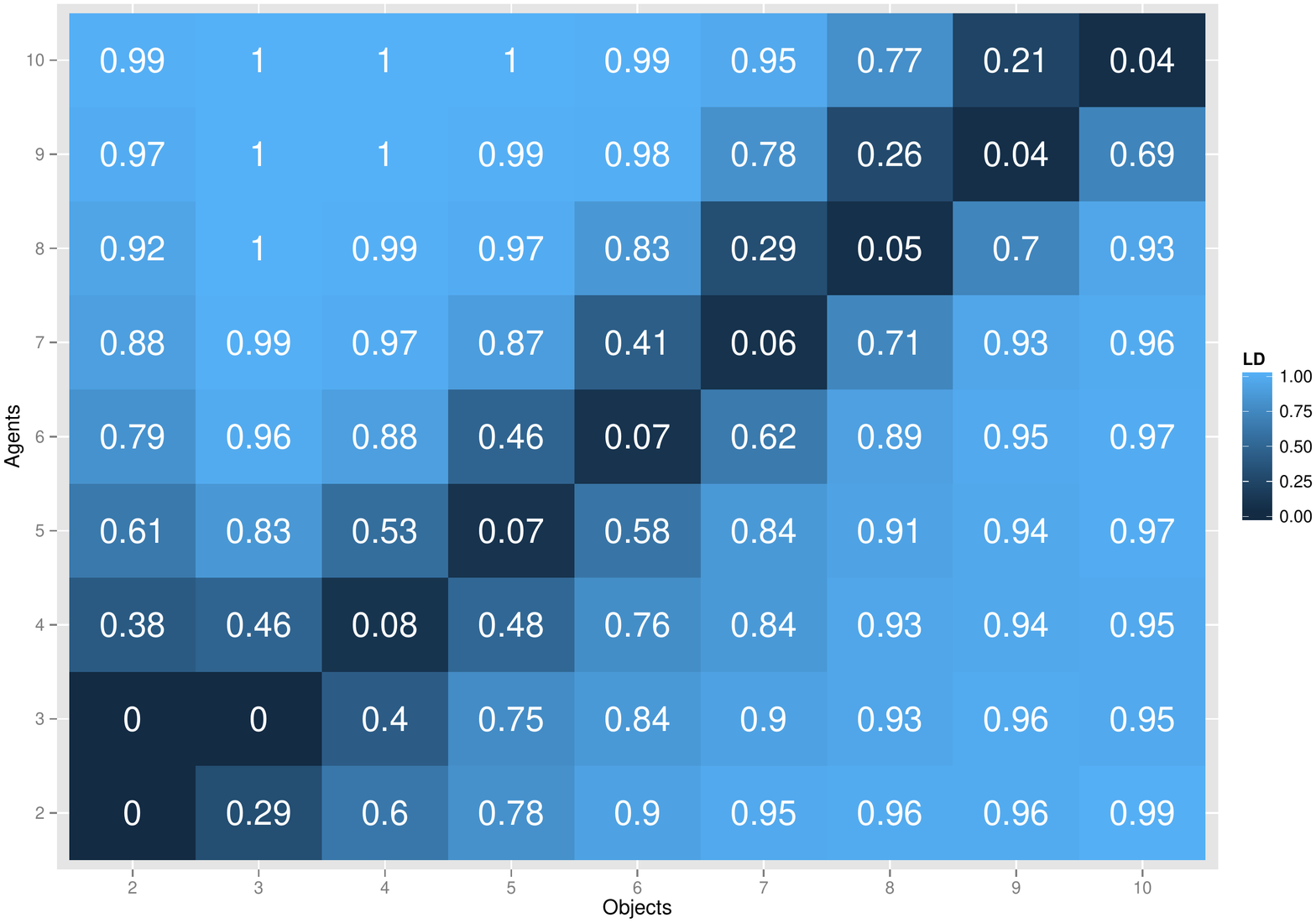}
                \caption{The fraction that PS lexicographically dominates RSD.}
                \label{fig:ld-efficiency}
\end{subfigure}
\caption{The fraction of preference profiles under which PS is superior to RSD.}
\label{Fig:efficiency}
\end{figure*}

Our first finding is that the fraction of preference profiles at which RSD and PS induce equivalent random assignments goes to $0$ when $n$ grows. 
There are two conclusions that one can draw from this observation. First, this result  confirms the theoretical results of Manea on asymptotic inefficiency of RSD~\cite{manea2009asymptotic}, in that, in most instances, the assignments induced by RSD are not equivalent to the PS assignments.
Second, this result also suggests that the incomparability of outcomes is significant, that is, the (ex ante) efficiency of the random outcomes is highly dependent on the underlying utility models. We propose the following conjecture.

%This further confirms the theoretical results of  Manea on asymptotic inefficiency of RSD~\cite{manea2009asymptotic}.
%However, based on observing the random assignments induced by PS and RSD, we conjecture that in a significant fraction of preference profiles the RSD and PS are in fact incomparable, meaning that the (ex ante) efficiency of the random outcome is highly dependent on the underlying utility models.

\begin{conjecture}\label{con:SDtoZero}
The fraction of preference profiles $\succ\in \mathcal{P}^{n}$ for which RSD is stochastically dominated by PS at $\succ$ converges to zero as $\frac{n}{m} \to 1$.
\end{conjecture}

Our empirical results support Conjecture~\ref{con:SDtoZero} on comparability of RSD and PS. Figure~\ref{fig:sd-efficiency} shows that although RSD is inefficient when $n$ grows beyond $n > 5$, due to incomparability of RSD and PS with regard to stochastic dominance relation, the RSD induced assignments are not stochastically dominated by \emph{sd}-efficient assignments induced by PS.

We also see similar results when we restrict ourselves to lexicographic preferences (Figure~\ref{fig:ld-efficiency}).

%Mihai Manea's proof on the inefficiency of RSD relies on Bogomolnaia and Moulin's lemma

%The theoretical proof of this conjecture is still an open problem, which requires a subtle characterization of the preference profiles, determining necessary conditions for stochastic dominance, and non-trivial combinatorial arguments on the space of preference profiles.

%Similarly, as shown in Figure~\ref{fig:ld-efficiency}, in the lexicographic domain when the number of agents and objects get close, the efficiency of PS and RSD is highly dependent on agents' utility models.

%the fraction of preference profiles for which PS assignment dominates the assignment by RSD converges to zero when the number of agents and objects get close.

\begin{conjecture}\label{con:LDtoZero}
The fraction of preference profiles $\succ\in \mathcal{P}^{n}$ for which RSD is lexicographically dominated by PS at $\succ$ converges to zero as $\frac{n}{m} \to 1$.
\end{conjecture}

For lexicographic preferences, we also observe that the fraction of preference profiles for which PS assignments strictly dominate RSD-induced allocations goes to 1 when the number of agents and objects diverge.

\begin{conjecture}\label{con:lex-dominated}
The fraction of preference profiles $\succ\in \mathcal{P}^{n}$ for which RSD is lexicographically dominated by PS at $\succ$ converges to 1 as $|n - m|$ grows.
\end{conjecture}

One immediate conclusion is that although RSD does not guarantee either \emph{sd}-efficiency or \emph{ld}-efficiency, in most settings with $n \leq m$, especially when $n=m$, neither of the two mechanisms is preferred in terms of efficiency. Hence, one cannot simply rule out the RSD mechanism.

\begin{figure*}[h]
\centering
\begin{subfigure}[b]{0.49\textwidth}
				\includegraphics[width=\textwidth]{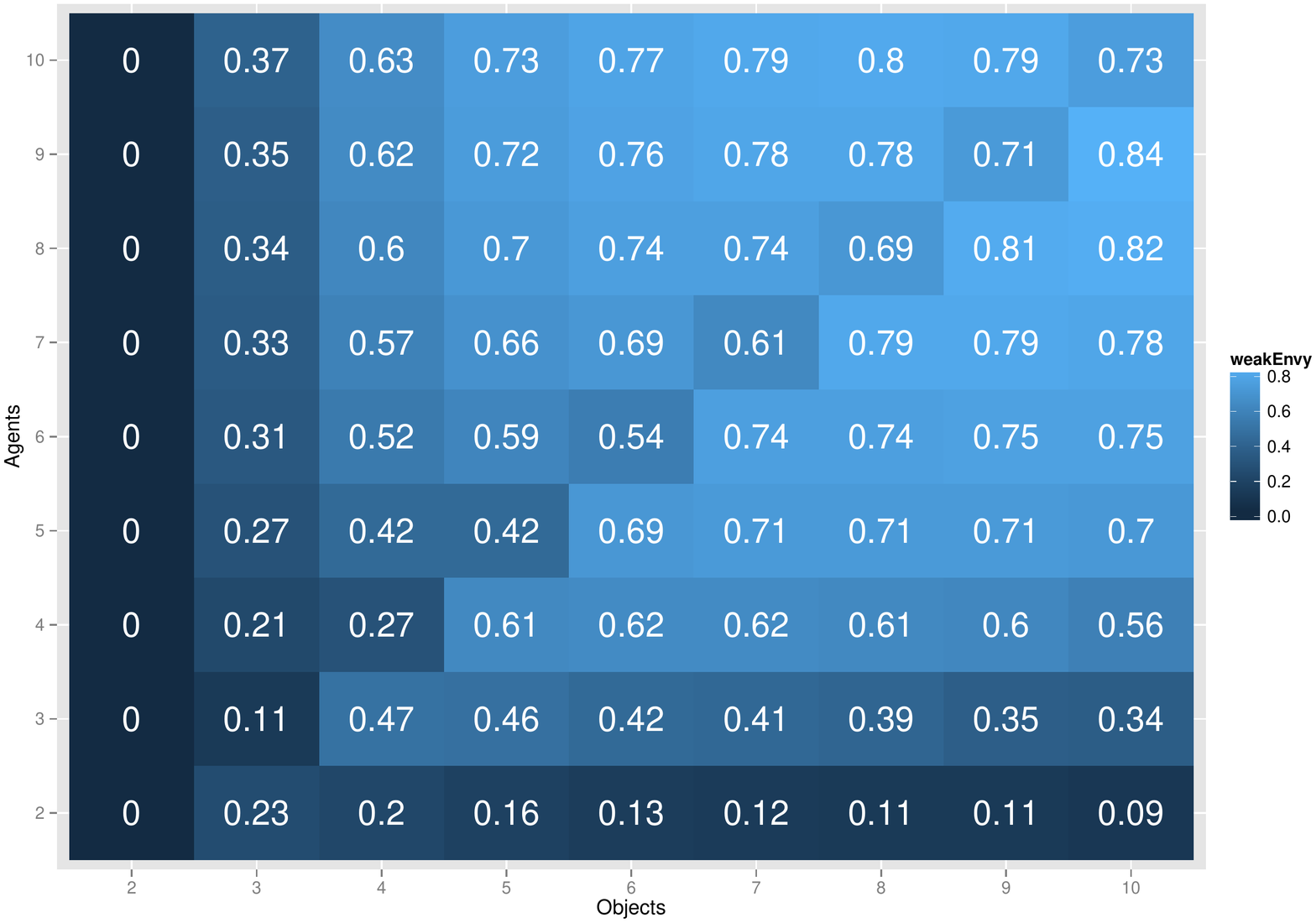}
                \caption{A heatmap showing the percentage of envious agents.}
                \label{Fig:a}
\end{subfigure}~
\begin{subfigure}[b]{0.49\textwidth}
\centering
\includegraphics[width=\linewidth]{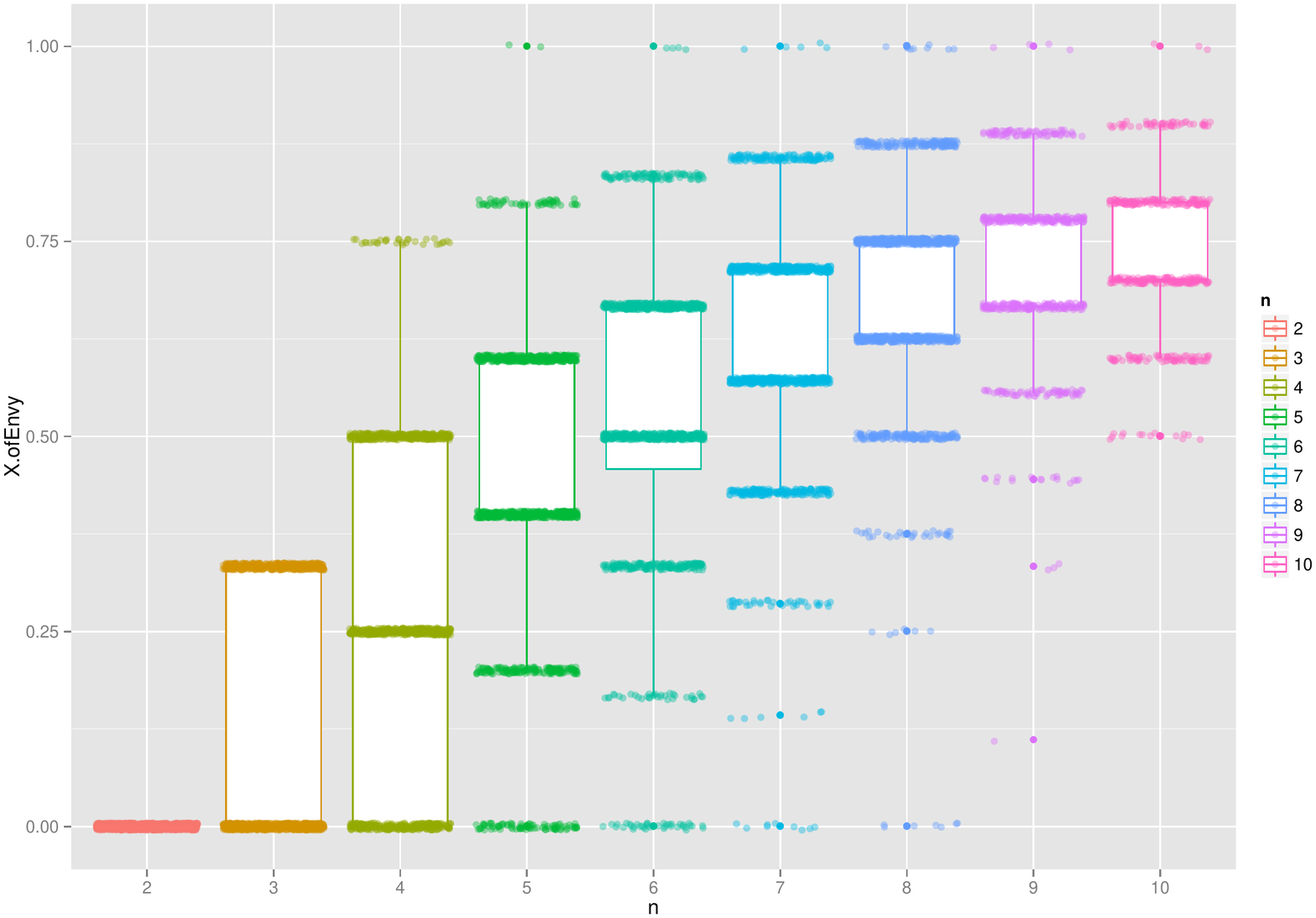}
\caption{Boxplots showing the various envy profiles for $n = m$. The Y axis represents the percentage of envious agents.}
\label{Fig:EnvyEqual}
\end{subfigure}
\caption{Plots representing the percentage of (weakly) envious agents under RSD.}
\label{Fig:RSD-Envy}
\end{figure*}
\begin{figure*}
\centering
\begin{subfigure}[b]{0.5\textwidth}
				\centering
				\includegraphics[width=\textwidth]{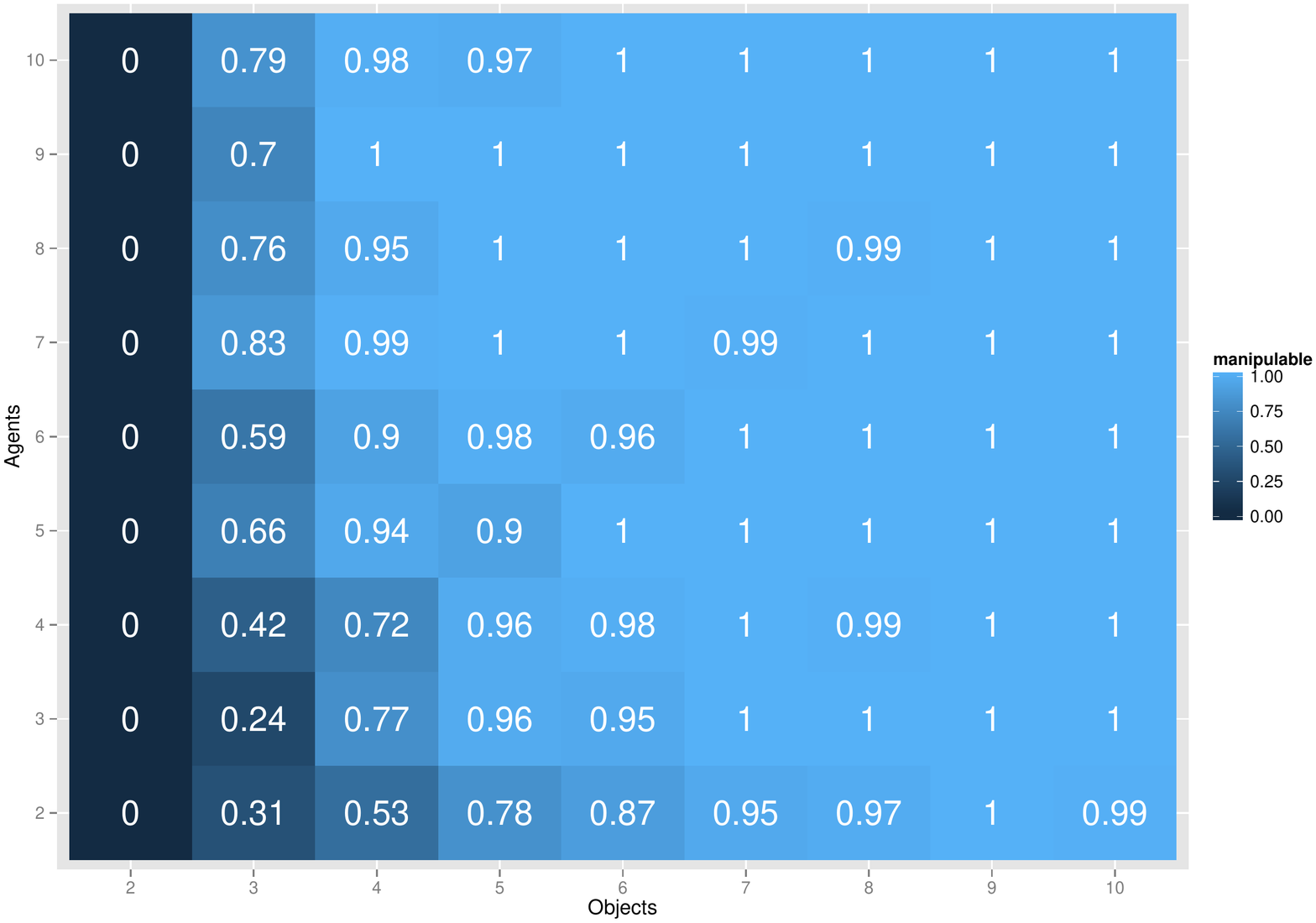}
                \caption{The fraction of manipulable preference profiles under PS.}
                \label{fig:manip}
\end{subfigure}~
\begin{subfigure}[b]{0.5\textwidth}
				\includegraphics[width=\textwidth]{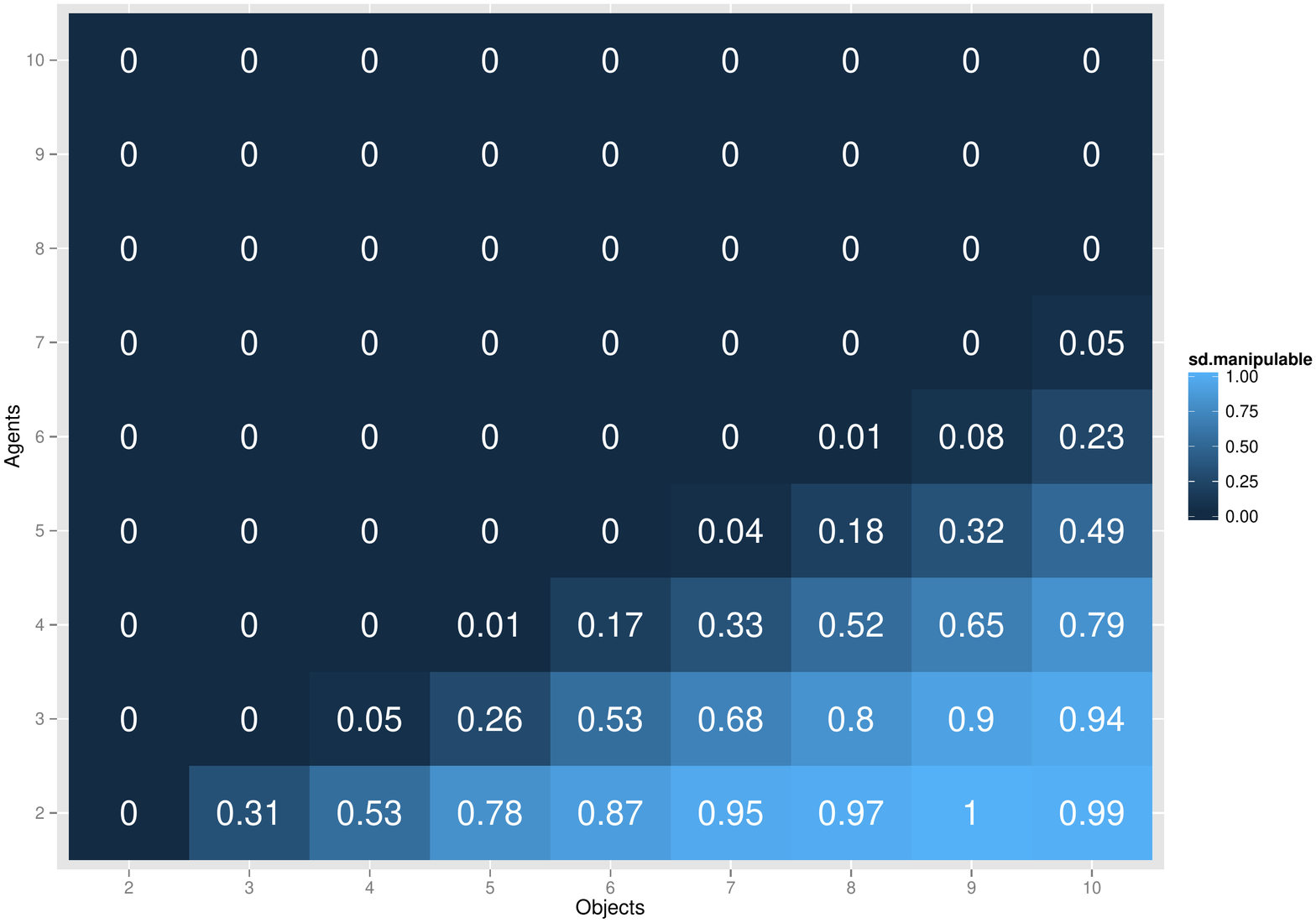}
                \caption{The fraction of \emph{sd}-manipulable preference profiles under PS.}
                \label{fig:sd-manipulable}
\end{subfigure}
\caption{Heatmaps illustrating the manipulablity of PS.} 
\label{Fig:manipulability}
\end{figure*}

\subsection{Envy in RSD}

%In the absence of well-defined utility models, we cannot evaluate the amount of weak envy induced by the RSD assignment. Thus, 

In our next experiment we measure the fraction of agents that are weakly envious of at least one another agent. Each data point represents the fraction of (weakly) envious agents. 

%For each $n$ and $m$, we uniformly sample from the entire space of preference profiles.

Figure \ref{Fig:RSD-Envy} shows that for RSD, the percentage of agents that are weakly envious increases with the number of agents. 
Figure~\ref{Fig:a} reveals an interesting observation: fixing any $n>3$, the percentage of agents that are (weakly) envious grows with the number of objects, however, there is a sudden drop in the percentage of envious agents when there are equal number of agents and objects. %$n=m$. Figure~\ref{Fig:b} further illustrates this observation on the envy of agents under the RSD mechanism.

For better understanding of the population of agents who feel (weakly) envious under RSD, we illustrate the envy distribution over the set of preference profiles for each $n=m$ (Figure \ref{Fig:EnvyEqual}). One observation is that there are few distinct envy profiles at each $n$, each representing a particular class of preference profiles, and by increasing $n$, the fraction of agents that are envious of at least one other agent increases.%, while the fraction of \emph{sd}-envyfree preference profiles converges to zero.

%Moreover, with larger number of agents the fraction of envious agents increases while the variance of envious agents decreases. 

\subsection{Manipulability of PS}

In Section~\ref{sec:lex}, we characterized the incentive properties of the PS mechanism with regards to the stochastic dominance and lexicographic dominance relations, arguing that although for $n\geq m$ PS is weakly \emph{sd}-strategyproof and \emph{ld}-strategyproof, when $n < m$ PS no longer satisfies these two properties.\footnote{A recent experimental study on the incentive properties of PS shows that human subjects are less likely to manipulate the mechanism when misreporting is a Nash equilibrium. However,  subjects' tendency for misreporting is still significant even when it does not improve their allocations~\cite{hugh2013experimental}. Hence, the PS mechanism suffers from incentive properties. In fact other mechanisms, such as the Draft mechanism, has shown to be highly susceptible to manipulation in many real-life markets such as course allocation at business schools~\cite{budish2012multi}.}
Nonetheless, we are interested in gaining insights on the fraction of preference profiles at which PS is manipulable.

%with this characterization we are still unable to postulate about the  
%Nonetheless, with the lack of strategyproofness (weak and strict) we are still unable to postulate 

%Our experiments on the space of preference profiles reveals a few subtle facts about the manipulability of the PS mechanism. 
Figure~\ref{fig:manip} shows that the fraction of manipulable preference profiles goes to 1 as $n$ or $m$ grow. In fact, PS is almost 100\% manipulable for $n > 5, m > 5$ (obviously PS is \emph{sd}-strategyproof when $m=2$.). 
Another compelling observation is that for all $n < m$, the fraction of \emph{sd}-manipulable preference profiles goes to 1 as $m - n$ grows (Figure~\ref{fig:sd-manipulable}). This states that in problems where the number of objects are larger than the number of agents, an agent can strictly benefit from misreporting her preferences. 
Figure~\ref{fig:ld-manipulable} shows that under lexicographic preferences, the fraction of \emph{ld}-manipulable preference profiles converges to 1 even more rapidly.

%Moreover, at those combinations of $n$ and $m$ where PS is sd-strategyproof for a fraction of possible preference profiles, the assignment induced by PS most often coincides with the RSD induced assignment. 
%For example, in Table \ref{Table:RSDvsPS} when $n=m=5$, PS is only sd-strategyproof at $11\% (100-89)$ of preference profiles, $7\%$ of which are identical to the assignment induced by RSD. This insight further confirms the vulnerability of PS to misreporting.

\begin{figure}%[b]{0.48\textwidth}
\centering
\includegraphics[width=.6\linewidth]{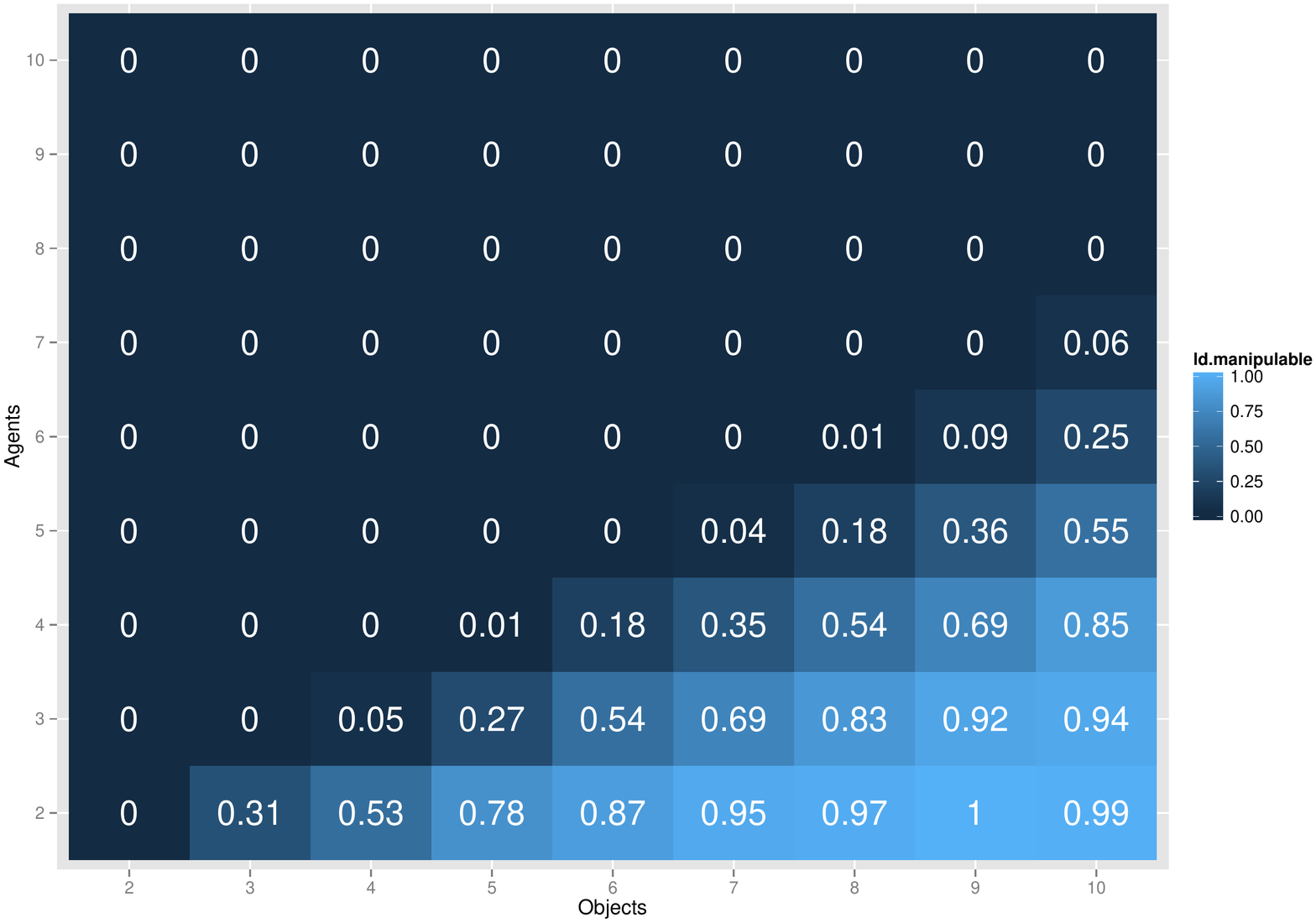}
%\vspace{-2em}
\caption{The fraction of \emph{ld}-manipulable profiles under PS.}
\label{fig:ld-manipulable}
\end{figure}

%%%%%%%%%%%%%%%%%%%%%%%%%%%%%%%%%%%%%%%%%%%%%%
%\section{Cardinal Utilities}
%%%%%%%%%%%%%%%%%%%%%%%%%%%%%%%%%%%%%%%%%%%%%%
%Here goes the experimental results for the Borda model.
%
%If time permits: experiments with Mallow models with various reference rankings and dispersion factors

%%%%%%%%%%%%%%%%%%%%%%%%%%%%%%%%%%%%%%%%%%%%%
\section{Discussion}

In this paper we conducted a comparison study of RSD and PS in order to gain further insight into their respective properties. We argued that outcomes generated by the mechanisms can be incomparable, making it difficult to determine which is the better mechanism in different settings. We then looked at some theoretical properties for the mechanisms under the assumption that agents have lexicographic preferences. Finally, we conducted an empirical comparison on the mechanisms. 

Earlier work had shown that RSD is \emph{sd}-strategyproof, ex post efficient and guarantees weak \emph{sd}-envyfreeness. We were able to further the theoretical understanding of the mechanism by showing that it is also \emph{ld}-envyfree. Furthermore, our empirical work showed that the actual fraction of outcomes where PS stochastically dominates RSD goes to 0 as $\frac{n}{m}$ goes to 1, raising the question as to whether the lack of \emph{sd}-efficiency guarantees for RSD is a significant concern in practice. 
In contrast, PS is \emph{sd}-efficient, \emph{sd}-envyfree and \emph{ld}-strategyproof for $n \geq m$. However, when $n < m$, the fraction of \emph{sd}- and \emph{ld}-manipulable assignments approaches 1, limiting the effectiveness of PS. 

Given the theoretical and empirical results, we suggest that in many cases RSD is, perhaps, a more suitable matching mechanism to consider, especially when $n \leq m$. This confirms previous studies on settings where there was an equal number of agents and objects~\cite{papai2001strategyproof,ehlers2003coalitional,hatfield2009strategy}.
We further strengthen this argument by providing an observation:
At preference profiles where PS and RSD induce identical assignments, RSD is \emph{sd}-efficient, \emph{sd}-envyfree, and \emph{sd}-strategyproof. However, PS may still be manipulable. 
\begin{example}
Consider the following preference profile $\succ = ((bca), (cab), (bca))$. Table~\ref{tab:equal} shows the induced random assignment.
\begin{table}
%\scriptsize
\small
\centering
 \begin{tabular}{ c c c c}
    \hline
     & $a$ & $b$ & $c$ \\ \hline\hline
    $A_1$ & $1/3$ & $1/2$ & $1/6$ \\
    $A_2$ & $1/3$ & $0$ & $2/3$ \\
    $A_3$ & $1/3$ & $1/2$ & $1/6$ \\
    \hline
  \end{tabular}%\vspace{-.5em}
  \caption{Assignment under $PS(\succ) = RSD(\succ)$}
  \label{tab:equal}
\end{table}
In this case, with PS as the matching mechanism, agent 1 can misreport her preference as $\succ'_{1}=(cba)$, and manipulate her assignment to $1/4(b), 1/2(c), 1/4(a)$. It is easy to see that agent 1's misreport improves her expected outcome for all utility models where $\frac{2}{6}u_{1}(c) > \frac{1}{4}u_{1}(b) + \frac{1}{12}u_{1}(a)$ (for example utilities $10, 9, 0$ for $b,c,a$ respectively.).
\end{example}

This work raises several interesting questions. From the theoretical perspective, among these open problems, providing theoretical proofs for the efficiency relations in Conjectures~\ref{con:SDtoZero}, \ref{con:LDtoZero}, and \ref{con:lex-dominated} is crucial and not trivial.
From the practical view, an interesting open problem is to see if there exists any randomized \emph{ld}-efficient mechanism that satisfies strategyproofness and some notion of fairness (such as proportionality) for $n < m$. Lastly, it would be interesting to investigate if there is an extension to the PS mechanism (or its variant based on different eating speeds) that is potentially weakly \emph{sd}-strategyproof for assigning multiple objects to each agent for $n < m$.

%%%%%%%%%%%%%%%%%%%%%%%%%%%%%%%%%%%%%%%%%%
%\clearpage

\bibliographystyle{plain}
\bibliography{ref}

\begin{thebibliography}{10}

\bibitem{abdulkadiroglu2009strategy}
Atila Abdulkadiroglu, Parag~A Pathak, and Alvin~E Roth.
\newblock Strategy-proofness versus efficiency in matching with indifferences:
  redesigning the new york city high school match.
\newblock Technical report, National Bureau of Economic Research, 2009.

\bibitem{abdulkadirouglu1998random}
Atila Abdulkadiro{\u{g}}lu and Tayfun S{\"o}nmez.
\newblock Random serial dictatorship and the core from random endowments in
  house allocation problems.
\newblock {\em Econometrica}, 66(3):689--701, 1998.

\bibitem{ashlagi2013mix}
Itai Ashlagi, Felix Fischer, Ian~A Kash, and Ariel~D Procaccia.
\newblock Mix and match: A strategyproof mechanism for multi-hospital kidney
  exchange.
\newblock {\em Games and Economic Behavior}, 2013.

\bibitem{Aziz:2014}
Haris Aziz, Serge Gaspers, Simon Mackenzie, and Toby Walsh.
\newblock Fair assignment of indivisible objects under ordinal preferences.
\newblock In {\em Proceedings of the 2014 International Conference on
  Autonomous Agents and Multi-agent Systems}, AAMAS '14, pages 1305--1312,
  Richland, SC, 2014. International Foundation for Autonomous Agents and
  Multiagent Systems.

\bibitem{bogomolnaia2012probabilistic}
Anna Bogomolnaia and Eun~Jeong Heo.
\newblock Probabilistic assignment of objects: Characterizing the serial rule.
\newblock {\em Journal of Economic Theory}, 147(5):2072--2082, 2012.

\bibitem{bogomolnaia2001new}
Anna Bogomolnaia and Herv{\'e} Moulin.
\newblock A new solution to the random assignment problem.
\newblock {\em Journal of Economic Theory}, 100(2):295--328, 2001.

\bibitem{bouveret2011general}
Sylvain Bouveret and J{\'e}r{\^o}me Lang.
\newblock A general elicitation-free protocol for allocating indivisible goods.
\newblock In {\em Proceedings of the Twenty-Second international joint
  conference on Artificial Intelligence-Volume Volume One}, pages 73--78. AAAI
  Press, 2011.

\bibitem{BL-ECAI14}
Sylvain Bouveret and Jérôme Lang.
\newblock Manipulating picking sequences.
\newblock In {\em Proceedings of the 21st European Conference on Artificial
  Intelligence (ECAI'14)}, Prague, Czech Republic, August 2014. IOS Press.

\bibitem{budish2012multi}
Eric Budish and Estelle Cantillon.
\newblock The multi-unit assignment problem: Theory and evidence from course
  allocation at harvard.
\newblock {\em American Economic Review}, 102(5):2237--71, 2012.

\bibitem{che2010asymptotic}
Yeon-Koo Che and Fuhito Kojima.
\newblock Asymptotic equivalence of probabilistic serial and random priority
  mechanisms.
\newblock {\em Econometrica}, 78(5):1625--1672, 2010.

\bibitem{domshlak2011preferences}
Carmel Domshlak, Eyke H{\"u}llermeier, Souhila Kaci, and Henri Prade.
\newblock Preferences in ai: An overview.
\newblock {\em Artificial Intelligence}, 175(7):1037--1052, 2011.

\bibitem{ehlers2003coalitional}
Lars Ehlers and Bettina Klaus.
\newblock Coalitional strategy-proof and resource-monotonic solutions for
  multiple assignment problems.
\newblock {\em Social Choice and Welfare}, 21(2):265--280, 2003.

\bibitem{fishburn1974lexicographic}
Peter~C Fishburn.
\newblock Lexicographic orders, utilities and decision rules: A survey.
\newblock {\em Management Science}, pages 1442--1471, 1974.

\bibitem{fraser1994ordinal}
Niall~M Fraser.
\newblock Ordinal preference representations.
\newblock {\em Theory and Decision}, 36(1):45--67, 1994.

\bibitem{hatfield2009strategy}
John~William Hatfield.
\newblock Strategy-proof, efficient, and nonbossy quota allocations.
\newblock {\em Social Choice and Welfare}, 33(3):505--515, 2009.

\bibitem{hugh2013experimental}
David Hugh-Jones, Morimitsu Kurino, and Christoph Vanberg.
\newblock An experimental study on the incentives of the probabilistic serial
  mechanism.
\newblock Technical report, Discussion Paper, Social Science Research Center
  Berlin (WZB), Research Area Markets and Politics, Research Unit Market
  Behavior, 2013.

\bibitem{kalinowski2013strategic}
Thomas Kalinowski, Nina Narodytska, Toby Walsh, and Lirong Xia.
\newblock Strategic behavior when allocating indivisible goods sequentially.
\newblock In {\em Twenty-Seventh AAAI Conference on Artificial Intelligence
  (AAAI-13)}, 2013.

\bibitem{kojima2009random}
Fuhito Kojima.
\newblock Random assignment of multiple indivisible objects.
\newblock {\em Mathematical Social Sciences}, 57(1):134--142, 2009.

\bibitem{kojima2010incentives}
Fuhito Kojima and Mihai Manea.
\newblock Incentives in the probabilistic serial mechanism.
\newblock {\em Journal of Economic Theory}, 145(1):106--123, 2010.

\bibitem{liu2013ordinal}
Qingmin Liu and Marek Pycia.
\newblock Ordinal efficiency, fairness, and incentives in large markets.
\newblock {\em Unpublished mimeo}, 2013.

\bibitem{manea2009asymptotic}
Mihai Manea.
\newblock Asymptotic ordinal inefficiency of random serial dictatorship.
\newblock {\em Theoretical Economics}, 4(2):165--197, 2009.

\bibitem{manlove2013algorithmics}
David Manlove.
\newblock {\em Algorithmics of matching under preferences}.
\newblock World Scientific Publishing, 2013.

\bibitem{papai2001strategyproof}
Szilvia P{\'a}pai.
\newblock Strategyproof and nonbossy multiple assignments.
\newblock {\em Journal of Public Economic Theory}, 3(3):257--271, 2001.

\bibitem{pathak2006lotteries}
Parag~A Pathak.
\newblock Lotteries in student assignment.
\newblock {\em Unpublished mimeo, Harvard University}, 2006.

\bibitem{roth2004kidney}
Alvin~E Roth, Tayfun S{\"o}nmez, and M~Utku {\"U}nver.
\newblock Kidney exchange.
\newblock {\em The Quarterly Journal of Economics}, 119(2):457--488, 2004.

\bibitem{saban2013note}
Daniela Saban and Jay Sethuraman.
\newblock A note on object allocation under lexicographic preferences.
\newblock {\em Journal of Mathematical Economics}, 2013.

\bibitem{schulman2012allocation}
Leonard~J Schulman and Vijay~V Vazirani.
\newblock Allocation of divisible goods under lexicographic preferences.
\newblock {\em arXiv preprint arXiv:1206.4366}, 2012.

\bibitem{sonmez2010course}
Tayfun S{\"o}nmez and M~Utku {\"U}nver.
\newblock Course bidding at business schools.
\newblock {\em International Economic Review}, 51(1):99--123, 2010.

\bibitem{sonmez2010house}
Tayfun S{\"o}nmez and M~Utku {\"U}nver.
\newblock House allocation with existing tenants: A characterization.
\newblock {\em Games and Economic Behavior}, 69(2):425--445, 2010.

\bibitem{svensson1999strategy}
Lars-Gunnar Svensson.
\newblock Strategy-proof allocation of indivisible goods.
\newblock {\em Social Choice and Welfare}, 16(4):557--567, 1999.

\bibitem{von1953certain}
John Von~Neumann.
\newblock A certain zero-sum two-person game equivalent to the optimal
  assignment problem.
\newblock {\em Contributions to the Theory of Games}, 2:5--12, 1953.

\bibitem{yaman2008democratic}
Fusun Yaman, Thomas~J Walsh, Michael~L Littman, and Marie Desjardins.
\newblock Democratic approximation of lexicographic preference models.
\newblock In {\em Proceedings of the 25th international conference on Machine
  learning}, pages 1200--1207. ACM, 2008.

\end{thebibliography}

\end{document}